\documentclass[11pt,uplatex]{article}
\usepackage{graphicx}
\usepackage{color}
\usepackage{amsmath, amssymb, amsthm, ascmac, caption, comment, enumerate, ascmac, bm}
\usepackage{latexsym, mathabx, mathrsfs, mathtools, psfrag, setspace, subcaption, booktabs, soul}
\usepackage{natbib}
\definecolor{mycolor}{cmyk}{0.80, 0.20, 0.25, 0}
\usepackage[colorlinks=true, linkcolor=mycolor, filecolor=mycolor, urlcolor=mycolor, citecolor=mycolor, driverfallback=dvipdfmx]{hyperref}
\usepackage[margin = 2.2cm]{geometry}
\mathtoolsset{showonlyrefs=true}

\usepackage{titlesec}
\titleformat*{\section}{\Large\bfseries\sffamily}
\titleformat*{\subsection}{\large\bfseries\sffamily}
\titleformat*{\subsubsection}{\large\bfseries\sffamily}

\DeclareMathOperator*{\argmin}{argmin}

\renewcommand{\hat}{\widehat}
\renewcommand{\tilde}{\widetilde}
\renewcommand{\check}{\widecheck}
\renewcommand{\bar}{\overline}
\newcommand{\bbE}{\mathbb{E}}

\newcommand{\bbR}{\mathbb{R}}

\numberwithin{equation}{section}
 
\allowdisplaybreaks[2]

\theoremstyle{definition}
\newtheorem{theorem}{Theorem}[section]
\newtheorem{assumption}{Assumption}[section]
\newtheorem{definition}{Definition}[section]

\newtheorem{lemma}{Lemma}[section]
\newtheorem{proposition}{Proposition}[section]
\newtheorem{remark}{Remark}[section]

\title{Functional Spatial Autoregressive Models}

\author{Tadao Hoshino\thanks{School of Political Science and Economics, Waseda University, 1-6-1 Nishi-waseda, Shinjuku-ku, Tokyo 169-8050, Japan. Email: \href{mailto:thoshino@waseda.jp}{thoshino@waseda.jp}.}}

\begin{document}
\maketitle

\begin{abstract}
    This study introduces a novel spatial autoregressive model in which the dependent variable is a function that may exhibit functional autocorrelation with the outcome functions of nearby units.
    This model can be characterized as a simultaneous integral equation system, which, in general, does not necessarily have a unique solution.
    For this issue, we provide a simple condition on the magnitude of the spatial interaction to ensure the uniqueness in data realization.
    For estimation, to account for the endogeneity caused by the spatial interaction, we propose a regularized two-stage least squares estimator based on a basis approximation for the functional parameter. 
    The asymptotic properties of the estimator including the consistency and asymptotic normality are investigated under certain conditions.
    Additionally, we propose a simple Wald-type test for detecting the presence of spatial effects.
    As an empirical illustration, we apply the proposed model and method to analyze age distributions in Japanese cities.
\end{abstract}

\newpage

\section{Introduction}

Spatial interdependence among units is an essential element in spatial data analysis.
To incorporate spatial interactions into econometric analysis, researchers have extensively utilized the Spatial Auto-Regressive (SAR) model:
\begin{align}\label{eq:SAR}
	y_i = \alpha_0 \sum_{j=1}^n w_{i,j} y_j + x_i^\top\beta_0 + \varepsilon_i,
\end{align}
where $y_i$ denotes a scalar outcome, $w_{i,j}$ denotes a known \textit{spatial weight} between $i$ and $j$, $x_i$ denotes a vector of explanatory variables, and $\varepsilon_i$ denotes an error term. 
The \textit{spatial lag} term $\sum_{j=1}^n w_{i,j} y_j$ captures the spatial trend of the outcome variable in the neighborhood of $i$, and the scalar parameter $\alpha_0$ measures its impact.
The usefulness of SAR modelling \eqref{eq:SAR} has been demonstrated in various empirical topics, including regional economics, local politics, real estate, crimes, etc.
In addition, if we define the weight term $w_{i,j}$ based on social distance or friendship connections instead of geographic distance, then the SAR models can be utilized to analyze social network data, and their applicability is vast.

To further broaden the applications of SAR modelling, this study aims to extend \eqref{eq:SAR} to a functional SAR model where the dependent variable is a function defined on a common closed interval:
\begin{align}\label{eq:model}
    q_i(s) = \sum_{j=1}^n w_{i,j} \int_0^1  q_j(t) \alpha_0(t,s) \text{d}t + x_i^\top \beta_0(s) + \varepsilon_i(s), \;\; \text{for} \;\; s \in [0,1],
\end{align}
where $q_i: [0,1] \to \mathbb{R}$ denotes the outcome function of interest.
Restricting the support to $[0,1]$ is a normalization.
In particular, for empirical relevance, this study primarily focuses on the case in which $q_i$ is the quantile function for a scalar dependent variable of interest.
Regression models involving functional variables have been widely studied in the literature of functional data analysis (FDA) for several decades (e.g., \citealp{ramsay2005functional}).
Our model is essentially different from the existing ones in that we explicitly consider the simultaneous spatial interactions of the outcome functions.

As a motivating example, suppose we intend to investigate the impact of a regional childcare subsidy program in a given city on the age distribution of the city.
The policy is likely to attract households with young children from other regions to benefit from the subsidy.
Additionally, if childcare facilities and schools need to be newly constructed, inflows of other age groups can also be anticipated as workers.
To obtain a comprehensive picture of the shift in the age distribution owing to the subsidy program in its entirety, it would be natural to consider a regression model in which the dependent variable represents the age distribution of each city, such as the quantile function.
Meanwhile, when the size of the young population in a given city is in an increasing trend (no matter the cause), which serves as a driver of economic growth of the city, this might also lead to an influx of working-age population into the surrounding regions owing to the spatial spillover of economic activities.
The proposed functional SAR model \eqref{eq:model} is able to account for such interdependency between the outcome functions of nearby spatial units.

In the literature, we are not the first to consider an SAR-type modelling in the functional regression context.
\cite{zhu2022network} proposed a social network model similar to ours in a time-series setting, where the response variable is a function of time.
They assumed that only concurrent interactions exist at each moment such that the past and future outcomes of others do not affect the present outcome.
Consequently, when fixed at each time point, their model can be reduced to the standard SAR model in \eqref{eq:SAR}.
In this regard, our model may be considered to be a generalization of theirs such that $\alpha_0(t,s) \neq 0$ is allowed for $t \neq s$ in general.

Another related modelling approach to ours is the SAR quantile regression (e.g., \citealp{su2011instrumental, malikov2019under, ando2023spatial}).
When $q_i$ represents a quantile function, our model and theirs are conceptually similar in that both approaches can examine the distributional effects of explanatory variables on the outcome and the spatial interaction of outcomes in a unified framework.
However, a fundamental distinction lies in that we consider a model in which each unit has its own unique quantile function as the dependent variable.
Consequently, we can explicitly allow for each specific quantile value of an outcome to interact with other quantiles of others' outcomes.
For instance, our model can investigate the impacts of median outcome of neighborhoods on a specific (say) 10 percentile value of own outcome.
In the time-series context, \cite{dong2024functional} consider the same type of interaction structure as above.

Notice that our model \eqref{eq:model} is characterized as a simultaneous integral equation system, and to the best of our knowledge, this type of modelling has not been investigated in the econometrics literature. 
To construct a consistent estimator for our model, the model space should be restricted such that the realized $q_i$'s are uniquely (in some sense) associated with the true parameters.
We show that to establish this uniqueness property, as in the standard SAR model (cf. \citealp{kelejian2010specification}), the spatial effects $\alpha_0$ must be bounded within a certain range.
In particular, we demonstrate that the tightness of the bound required for $\alpha_0$ depends on the smoothness of the outcome function.

To estimate the model parameters, we need to address the endogeneity issue arising from the simultaneous interaction among the outcome functions.
Thus, we propose a regularized two-stage least squares (2SLS) estimator that is based on a series approximation of $\alpha_0(\cdot, s)$ at each evaluation point $s$.
Under the availability of a sufficient number of instrumental variables (IVs) and regularity conditions, we prove that both the estimator for $\beta_0$ and that for $\alpha_0(t,s)$ are consistent at certain convergence rates and asymptotically normally distributed.
Additionally, we develop a Wald-type test for assessing the presence of any spatial effects at each $s$.
We show that the proposed test statistic asymptotically distributes as the standard normal after appropriate normalization.
Furthermore, we discuss performing the estimation when the outcome functions are not fully observable on the entire interval $[0,1]$, but are only discretely observed, which is typical in most empirical situations.
Our proposed estimator relies on a simple interpolation method, and we derive a set of conditions under which the estimator can achieve the same asymptotic properties as the infeasible counterpart.

As an empirical illustration, we investigate the determinants of age distribution in Japanese cities.
Since many Japanese cities are currently rapidly aging, which has emerged as one of the central social problems in the country, understanding the mechanisms underlying the age structure of cities is crucial.
Using recent government survey data, including the Census, we apply our estimation and testing method to 1883 Japanese cities.
Here, the outcome function $q_i$ represents the quantile function of the age distribution in city $i$, and covariates $x_i$ include variables such as annual commercial sales, unemployment rate, number of childcare facilities, and others.
Our results suggest that spatial interaction effects are extremely weak at quantiles close to the boundary points 0 or 1.
This may not be surprising as all individuals are born at age 0 and have a life expectancy of approximately 100 years at maximum, resulting in little regional heterogeneity.
In contrast, strong spatial effects are observed when both $t$ and $s$ are at approximately the ages of young working population, possibly indicating that economic activities and their spillovers are the main factors in shaping the spatial trend of age structure.

The remainder of this paper is organized as follows:
In Section \ref{sec:model}, we formally introduce the model proposed in this study and discuss the condition under which it is well defined with a unique solution.
In addition, focusing on the cases where the outcome function is a quantile function, we discuss the motivations and interpretation of such a modelling approach.
In Section \ref{sec:estimator}, we describe our 2SLS method for estimating $\beta_0$ and $\alpha_0$.
Thereafter, we study the asymptotic properties of the proposed estimator under a set of assumptions.
In this section, we also propose a test statistic for testing the null hypothesis that $\alpha_0(t, s) = 0$ for $t \in \mathcal{I}$, and its asymptotic distribution is derived.
In Section \ref{sec:MC}, we present the results of Monte Carlo experiments to evaluate the finite sample performance of the proposed estimator and test.
Section \ref{sec:empir} presents our empirical analysis on the age distribution of Japanese cities, and Section \ref{sec:conclude} concludes the paper.

\paragraph{Notation}

For a natural number $n$, $I_n$ denotes an $n \times n$ identity matrix.
For a function $h$ defined on $[0,1]$, the $L^p$ norm of $h$ is written as $||h||_{L^p} \coloneqq (\int_0^1 |h(s)|^p \text{d}s)^{1/p}$, and $L^p(0,1)$ denotes the set of $h$'s such that $||h||_{L^p} < \infty$.
For a random variable $x$, the $L^p$ norm of $x$ is written as $||x||_p \coloneqq (\bbE|x|^p)^{1/p}$.
For a matrix $A$, $|| A ||$ and $||A||_\infty$ denote the Frobenius norm and the maximum absolute row sum of $A$, respectively.
If $A$ is a square matrix, we use $\rho_{\max} (A)$ and $\rho_{\min} (A)$ to denote its largest and smallest eigenvalues, respectively. 
In addition, $A^{-}$ is a symmetric generalized inverse of $A$.
We write $a \lesssim b$ and $a \lesssim_p b$ if $a = O(b)$ and $a = O_P(b)$, respectively.
Finally, we write $a \sim b$ when $a \lesssim b$ and $b \lesssim a$.

\section{Functional SAR Models}\label{sec:model}

\subsection{Model Setup and Completeness}

Suppose that we have data of size $n$: $\{(q_i, x_i, w_{i,1}, \ldots, w_{i,n})\}_{i = 1}^n$, where $q_i$ denotes a random outcome function of interest with the common support $[0,1]$, $x_i = (x_{i,1}, \ldots, x_{i,d_x})^\top \in \mathbb{R}^{d_x}$ denotes a vector of covariates including a constant term, and $w_{i,j} \in \mathbb{R}$ denotes the $(i,j)$-th element of an $n \times n$ pre-specified spatial weight matrix $W_n = (w_{i,j})_{i,j = 1}^n$.
The value of each $w_{i,j}$ is determined non-randomly.
As is the convention, we set $w_{i,i} = 0$ for all $i$ for normalization.
Note that the spatial configurations of the units generally change with the sample size.
Thus, the variables generally form triangular arrays, and model parameters depend on $n$ through spatial interactions.
However, when there is no confusion, the dependence on $n$ is suppressed for notational convenience.

As shown in \eqref{eq:model}, our working model is
\begin{align}
    q_i(s) = \int_0^1 \bar q_i(t) \alpha_0(t,s) \text{d}t + x_i^\top \beta_0(s) + \varepsilon_i(s), \;\; \text{for} \;\; s \in [0,1],
\end{align}
where $\bar q_i$ denotes the spatial lag of the outcome function: $\bar q_i \coloneqq \sum_{j = 1}^n w_{i,j} q_j$.
The unknown parameters to be estimated are $\alpha_0$ and $\beta_0 = (\beta_{0 1}, \ldots, \beta_{0 d_x})^\top$.
For instance, in our empirical analysis, $q_i(s)$ denotes the $s$-th quantile of the age distribution in city $i$, and $\alpha_0(t,s)$ captures the impacts from the $t$-th quantile ages of neighborhood cities to the $s$-th quantile age of own city.
For other examples, $q_i(s)$ could be the $s$-th quantile of the income distribution in city $i$, $s$-th quantile of the daily activity energy expenditure of person $i$, number of available bicycles at the bicycle-sharing station $i$ at time $s$, and so forth.
Hereinafter, we assume that $q_i \in L^p(0,1)$ for some $2 \le p < \infty$ and that $\alpha_0 \in C[0,1]^2$,  where $C[0,1]^2$ denotes the set of continuous functions on $[0,1]^2$.

\bigskip

Before turning to the estimation of $\alpha_0$ and $\beta_0$, we discuss the \textit{completeness} of our model, that is, whether model \eqref{eq:model} can be characterized by a unique solution $(q_1, \ldots, q_n)$.
As our model comprises a system of $n$ functional equations, the existence and uniqueness of the solution are non-trivial problems.
If the system does not have or has multiple solutions,
consistently estimating the model parameters without some ad hoc assumptions is generally impossible.

Let $Q(s) = (q_1(s), \ldots, q_n(s))^\top$, $X = (x_1, \ldots, x_n)^\top$, and $\mathcal{E}(s) = (\varepsilon_1(s), \ldots, \varepsilon_n(s))^\top$.
Then, we can re-write \eqref{eq:model} in matrix form as
\begin{align}
    Q(s) = W_n \int_0^1 Q(t) \alpha_0(t,s) \text{d}t + X \beta_0(s) + \mathcal{E}(s).
\end{align}
This expression suggests that our model is seen as a system of \textit{Fredholm integral equations of the second kind} with kernel $\alpha_0(t,s)$.
Defining $\bar \alpha_0 \coloneqq \max_{(t,s) \in [0,1]^2} |\alpha_0(t,s)|$, whose existence is ensured under the continuity of $\alpha_0$, assume the following:
\begin{assumption}\label{as:inverse}
    $\bar \alpha_0 \lesssim 1$ and $||W_n||_\infty \lesssim 1$ such that $\bar \alpha_0 ||W_n||_\infty < 1$.
\end{assumption}

Let us denote $\mathcal{H}_{n,p} \coloneqq \{H = (h_1, \ldots, h_n) : h_i \in L^p(0,1) \; \text{for all} \; i\}$, and define a linear operator $\mathcal{T}$ as
\begin{align}
    (\mathcal{T} H)(s) \coloneqq W_n \int_0^1 H(t) \alpha_0(t, s) \text{d}t, \;\; \text{for} \;\; H \in \mathcal{H}_{n,p},
\end{align}
whose range is $\mathcal{H}_{n,p}$ under Assumption \ref{as:inverse}.
Thus, we can write $Q = \mathcal{T} Q + X\beta_0 + \mathcal{E}$.
Then, denoting $\text{Id}$ to be the identity operator, if the inverse operator $(\text{Id} - \mathcal{T})^{-1}$ exists, the solution $Q$ of the system can be uniquely determined (as an element of $\mathcal{H}_{n,p}$) as $Q = (\text{Id} - \mathcal{T})^{-1}[X \beta_0 + \mathcal{E}]$.

The next proposition states that Assumption \ref{as:inverse} is sufficient for the existence of $(\text{Id} - \mathcal{T})^{-1}$ and uniqueness of $Q$.

\begin{proposition}\label{prop:completeness}
Suppose that Assumption \ref{as:inverse} holds.
Then, $(\text{Id} - \mathcal{T})^{-1}$ exists, and $Q$ is the only solution of \eqref{eq:model} in the Banach space $(\mathcal{H}_{n,p}, ||\cdot||_{\infty, p})$, where $||H||_{\infty, p} \coloneqq \max_{1 \le i \le n} ||h_i||_{L^p}$.
\end{proposition}

The proof is straightforward.
Under Assumption \ref{as:inverse}, we have
\begin{align}\label{eq:Lpinequality}
    \begin{split}
    \left\| \{\mathcal{T} H\}_i \right\|_{L^p}
     = \left\| \sum_{j = 1}^n w_{i,j} \int_0^1 h_j(t) \alpha_0(t, \cdot) \text{d}t \right\|_{L^p}
    & \le \sum_{j = 1}^n |w_{i,j}| \left( \int_0^1 \left| \int_0^1 h_j(t) \alpha_0(t, s) \text{d}t\right|^p \text{d}s \right)^{1/p} \\
    & \le \sum_{j = 1}^n |w_{i,j}| \left( \int_0^1 \int_0^1 \left| h_j(t) \right|^p \left|\alpha_0(t, s) \right|^p  \text{d}t \text{d}s \right)^{1/p} \\
    & \le \bar \alpha_0 ||W_n||_\infty \max_{1 \le j \le n}||h_j||_{L^p} < ||H||_{\infty, p} < \infty
    \end{split}
\end{align}
for any $H \in \mathcal{H}_{n,p}$ by Minkowski's and Jensen's inequalities.
This implies that $\mathcal{T}H \in \mathcal{H}_{n,p}$.
As is well known, if the operator norm of $\mathcal{T}$ is smaller than one, $(\text{Id} - \mathcal{T})^{-1}$ exists, and we have the Neumann series expansion $(\text{Id} - \mathcal{T})^{-1} = \sum_{\ell = 0}^\infty \mathcal{T}^\ell$ converging in the operator norm (e.g., Theorem 2.14, \cite{Kress2014linear}).
It is immediate from \eqref{eq:Lpinequality} that $\left\| \mathcal{T} H  \right\|_{\infty, p} < 1$ follows for any $H$ such that $||H||_{\infty, p} = 1$, which yields the desired result.

\bigskip

When the spatial weight matrix is row-normalized such that $||W_n||_\infty = 1$, as is often the case in empirical applications, Assumption \ref{as:inverse} can be reduced to $\bar \alpha_0  < 1$, which somewhat resembles the solvability condition $|\alpha_0| < 1$ for the standard linear SAR model \eqref{eq:SAR}.

\begin{remark}[Alternative condition]
    If one imposes a stronger assumption on the space of the input functions, the requirement for the kernel can be relaxed.
    For example, for all $i$, suppose that $q_i$ belongs to $C[0,1]$. 
    Then, by the extreme value theorem, $q_i$'s are bounded.
    Letting $\mathcal{H}_{n, \infty} \coloneqq \{H = (h_1, \ldots, h_n) : h_i \in C[0,1] \; \text{for all} \; i\}$ and $||H||_{\infty, \infty} \coloneqq \max_{1 \le i \le n} \max_{s \in [0,1]}|h_i(s)|$, we can easily show that $Q$ is the only solution in the Banach space $(\mathcal{H}_{n, \infty}, ||\cdot||_{\infty, \infty})$ if $||W_n||_\infty \max_{s \in [0,1]}  \int_0^1 |\alpha_0(t,s)| \text{d}t < 1$ is satisfied.\footnote{
        Clearly, for any given $H \in \mathcal{H}_{n,\infty}$ such that $||H||_{\infty, \infty} = 1$, we have
        \begin{align*}
            \left| \{(\mathcal{T} H)(s)\}_i \right|
            \le \sum_{j = 1}^n |w_{i,j}| \int_0^1 |h_j(t)| \cdot |\alpha_0(t,s)| \text{d}t 
            \le ||W_n||_\infty \max_{s \in [0,1]}  \int_0^1 |\alpha_0(t,s)| \text{d}t.
        \end{align*}
    }
    If the spatial weight matrix is row-normalized, then the condition can be further simplified to $\max_{s \in [0,1]} \int_0^1 |\alpha_0(t,s)| \text{d}t  < 1$, which is a familiar requirement for the solvability of the Fredholm integral equation of the second kind (e.g., Corollary 2.16, \cite{Kress2014linear}).
    It is known that compactly supported continuous functions are dense in $L^p$ ($1 \le p < \infty$). 
    Thus, in practice, assuming that all $q_i$'s are continuous is almost harmless, and hence the violation of Assumption \ref{as:inverse} should be allowed to some extent. 
\end{remark}

The Neumann series expansion implies that $Q$ can be expressed as
$Q = X\beta_0 +  \mathcal{T}X\beta_0 + \mathcal{T}^2 X\beta_0 + \cdots + \mathcal{E} + \mathcal{T}\mathcal{E} + \mathcal{T}^2 \mathcal{E} + \cdots$, that is,
\begin{align}
    Q(\cdot) 
    &= X\beta_0(\cdot) + W_n X \int_0^1 \beta_0(t) \alpha_0(t, \cdot) \text{d}t + W_n W_n X \int_0^1  \int_0^1 \beta_0(t_1) \alpha_0(t_1, t_2) \alpha_0(t_2, \cdot) \text{d}t_1 \text{d}t_2 + \cdots   
\end{align}
Hence, the marginal effect of increasing $x_{i,j}$ on $Q(\cdot)$ is obtained by
\begin{align}
    \frac{\partial Q(\cdot)}{\partial x_{i,j}}
    & = \bm{e}_i \beta_{0j}(\cdot) + W_n \bm{e}_i \int_0^1  \beta_{0j}(t) \alpha_0(t,\cdot) \text{d}t + W_n W_n \bm{e}_i \int_0^1 \int_0^1  \beta_{0j}(t_1) \alpha_0(t_1,t_2) \alpha_0(t_2, \cdot) \text{d}t_1 \text{d}t_2 + \cdots,
\end{align}
where $\bm{e}_i$ denotes the $i$-th column of $I_n$.
This clearly shows that a change in $i$'s covariate affects not only the outcome of $i$ but also those of other units through the spatial interaction - the so-called \textit{spatial multiplier} effect.

\subsection{Leading Example: A Distributional SAR Model}

One of the situations in which model \eqref{eq:model} can be most nicely applied empirically would be when the outcome function $q_i$ represents the quantile function for the cumulative distribution function (CDF) of a variable of interest.
In our empirical analysis, we study the determinants of the  \textit{population pyramids} of Japanese cities by employing the age quantile function of city $i$ as $q_i$.

Suppose that for each $i$ we can observe a random CDF $F_i$ for an outcome variable $y \in \mathcal{Y}_i \subseteq \mathbb{R}$ of interest.
The quantile function of $y$ for $i$ is defined as $q_i(s) \coloneqq \inf\{y \in \mathcal{Y}_i : s \le F_i(y) \}$.
In the FDA literature, models where the response variable represents a probability distribution have garnered significant attention, for example, \cite{petersen2016functional}, \cite{han2020additive}, \cite{yang2020quantile}, \cite{yang2020random}, \cite{ghodrati2022distribution}, \cite{petersen2021wasserstein}, \cite{chen2023wasserstein}.
For an excellent review on this topic, refer to \cite{petersen2022modeling}.
A common view in these studies is that performing a regression analysis directly in the space of CDFs (or densities) is often problematic.
Hence, we should consider imposing a regression model on the quantile function (rather than on the CDF per se), as in \cite{yang2020quantile} and \cite{yang2020random}, enabling us to enjoy several analytically and interpretationally preferable properties as mentioned below.

First, quantile functions can be easily computed without considering the range boundaries, unlike CDFs.
Second, the domains of CDFs are typically heterogeneous across individuals, whereas that of quantile functions is always the fixed interval $[0,1]$.
Third, the least-squares regression of the quantile function can be nicely interpreted as a \textit{Wasserstein distance} minimization problem.

More specifically for the third point, denoting $F_i^{\alpha,\beta}$ to be the CDF induced from the quantile function $q_i^{\alpha, \beta}(s) \coloneqq \int_0^1 \bar q_i(t) \alpha(t,s) \text{d}t + x_i^\top \beta(s)$, the squared $2$-Wasserstein distance between $F_i$ and $F_i^{\alpha, \beta}$ is obtained as\footnote{
    Formally, the Wasserstein distance is a distance between two probability measures.
    We abuse the notation using CDFs in its arguments for ease of explanation.
    For more precise discussions on the properties of Wasserstein distance, see \cite{panaretos2020invitation}, for instance.
}
\begin{align}
    \mathcal{W}_2^2(F_i, F_i^{\alpha,\beta}) = \int_0^1 \left( q_i(s) - q_i^{\alpha,\beta}(s) \right)^2\text{d}s.
\end{align}
Thus, minimizing the mean squared Wasserstein distance with respect to $(\alpha, \beta)$: $\min_{\alpha, \beta} n^{-1} \sum_{i = 1}^n \mathcal{W}_2^2(F_i, F_i^{\alpha,\beta})$ is equivalent to performing a functional least squares regression based on model \eqref{eq:model}.
Note, however, that the resulting least squares estimator does not produce a consistent estimate of $(\alpha_0, \beta_0)$ because of the endogeneity of $\bar q_i$.
To circumvent the endogeneity issue, we introduce a penalized 2SLS method in the next section.

It is also worth noting that the spatially lagged quantile $\bar q_i$ corresponds to the quantile function of the spatially weighted \textit{Fr\'{e}chet mean} $\bar F_i$ of $\{F_1, \ldots, F_n\}$ at the location of $i$: $\bar F_i \coloneqq \arg\min_F \sum_{j=1}^n w_{i,j} \mathcal{W}_2^2(F, F_j)$, which is also referred to as the \textit{Wasserstein barycenter}, with $w_{i,j} \ge 0$ and $\sum_{j=1}^n w_{i,j} = 1$.
Rather than using $\bar q_i$, one might consider employing the quantile function of the spatially lagged CDF: $\sum_{j=1}^n w_{i,j} F_i$ as the spatial trend term.
However, a linear mixture of CDFs is generally multimodal and does not inherit the shape properties of the original CDFs.
In this regard, the weighted Fr\'{e}chet mean should be more representative and faithful as an indicator of the neighborhood trend.
This point is also highlighted in \cite{gunsilius2023distributional} in the context of synthetic control analysis.

\section{Estimation and Asymptotics}\label{sec:estimator}

\subsection{Penalized 2SLS estimator}

We now discuss the estimation of $\alpha_0$ and $\beta_0$.
Let $\{\phi_k: k = 1,2, \ldots\}$ be a series of basis functions, such as Fourier series, B-splines, and wavelets, such that we can expand $\alpha_0(t,s) = \sum_{k = 1}^\infty \phi_k(t) \theta_{0k}(s)$ for each $s$.
Then, we have $\int_0^1 q_i(t) \alpha_0(t,s) \text{d} t = \sum_{k = 1}^\infty r_{i,k} \theta_{0k}(s)$, where $r_{i,k} \coloneqq \int_0^1 q_i(t) \phi_k(t) \text{d}t$.
Hence, our model \eqref{eq:model} can be re-written as
\begin{align}
    q_i(s) 
    & = \sum_{k = 1}^K \bar r_{i,k} \theta_{0k}(s) + x_i^\top \beta_0(s) + \varepsilon_i(s) + u_i(s)
\end{align}
where $\bar r_{i,k} \coloneqq \sum_{j=1}^n w_{i,j} r_{j,k}$, $u_i(s) \coloneqq \int_0^1 \bar q_i(t) \alpha_0(t,s) \text{d}t - \sum_{k = 1}^K \bar r_{i,k} \theta_{0k}(s)$, and $K \equiv K_n$ is a sequence of integers tending to infinity as $n$ grows.
Note that this is just a multiple regression model having $K$ endogenous regressors $\bar r_i = (\bar r_{i,1}, \ldots, \bar r_{i,K})^\top$ with a composite error term $ \varepsilon_i(s) + u_i(s)$.
Thus, we can resort to the 2SLS approach to estimate $\theta_0(s) = (\theta_{01}(s), \ldots, \theta_{0K}(s))^\top$ and $\beta_0(s)$ under the availability of a sufficient number of valid IVs for $\bar r_i$.

Suppose that we have an $L \times 1$ vector $z_{1,i}$ of IVs that are correlated with $\bar r_i$ but not with $\varepsilon_i$ such that $L \equiv L_n \ge K_n$.
The choice of IVs will be discussed later.
Further, let $z_i = (z_{1,i}^\top, x_i^\top)^\top$, $Z = (z_1, \ldots, z_n)^\top$, $\bar R = (\bar r_1, \ldots, \bar r_n)^\top$, $M_z = Z (Z^\top Z)^{-} Z^\top$, $M_x = X (X^\top X)^{-1} X^\top$, and $\bar R_x = (I_n - M_x) \bar R$.
Then, our 2SLS estimator is defined as follows:
\begin{align}\label{eq:2SLS}
    \begin{split}
        \hat \beta_n(s)
        & \coloneqq \left[ X^\top (I_n - S) X \right]^{-1} X^\top (I_n - S) Q(s)\\
        \hat \theta_n(s)
        & \coloneqq \left[ \bar R_x^\top M_z \bar R_x + \lambda D n \right]^{-1} \bar R_x^\top M_z Q(s) 
    \end{split}
\end{align}
for a given evaluation point $s \in (0,1)$, where $S = M_z \bar R [ \bar R^\top M_z \bar R ]^{-} \bar R^\top M_z$, $\lambda \equiv \lambda_n$ is a non-negative regularization parameter tending to zero as $n$ increases, and $D$ denotes a $K$-dimensional matrix, which is positive semidefinite, symmetric, and satisfies $\rho_\text{max} (D) \lesssim 1$ uniformly in $K$.
Once $\hat \theta_n(s) = (\hat \theta_{n1}(s), \ldots, \hat \theta_{nK}(s))^\top$ is obtained, we can estimate $\alpha_0(\cdot,s)$ by
\begin{align}
    \hat \alpha_n(\cdot,s) \coloneqq \sum_{k=1}^K \phi_k(\cdot) \hat \theta_{nk}(s).
\end{align}
To recover the entire functional form of $\alpha_0(\cdot, \cdot)$ and $\beta_0(\cdot)$, we can repeat the described estimation procedure over a sufficiently fine grid on $[0,1]$. 

\begin{remark}[Choice of instruments]\label{rem:IV}
    Observing that $\bar{q}_i(s) \approx \int_0^1 \bar{\bar{q}}_i(t) \alpha_0(t, s)\text{d}t + \bar{x}_i^\top \beta_0(s)$, where $\bar{\bar{q}}_i(t) \coloneqq \sum_{j=1}^n w_{i,j} \bar{q}_j(t)$, and $\bar{x}_i \coloneqq \sum_{j=1}^n w_{i,j} x_j$, the spatially lagged covariates $\bar{x}_i$ would be natural IV candidates for $\bar{r}_{i,k}  = \int_0^1 \bar{q}_i(t) \phi_k(t)\text{d}t$, assuming that $\beta_0 \neq 0$.
    For identification, the number of valid IVs must be larger than or equal to $K$.
    While it is theoretically required that $K$ tends to infinity as $n$ increases to consistently estimate $\alpha_0(\cdot, s)$, the dimension of $\bar{x}_i$, $d_x$, is fixed in our model. 
    Note that, as long as both $\alpha_0$ and $\beta_0$ are non-degenerate, it is possible to create arbitrarily many IVs by taking the spatial lags of $x_i$ of higher and higher order: $\bar{x}_i$, $\bar{\bar{x}}_i$, ... and so forth.
    However, the higher the order, the weaker the instruments.
    Since $K$ is at most less than eight or so for most practical sample sizes, we believe that finding sufficient IVs may not be a serious concern in most empirical situations where researchers can collect a reasonable number of independent variables.
    See Remark \ref{rem:ID} below for a related discussion.
\end{remark}

In practice, the 2SLS estimator in \eqref{eq:2SLS} would be rarely feasible because function $q_i$ can usually only be incompletely observed. 
For instance, we might only be able to observe the values of $q_i$ at finite points, ${q_i(s_{i,1}), \ldots, q_i(s_{i,m_i})}$.
This is the case of our empirical analysis of the age distribution in Japanese cities.
In this empirical analysis, we cannot access the complete age distribution for each city, but we only know the distribution up to every five-year age interval.
In such a case, for example, we can apply a linear interpolation method to obtain an approximation of the entire functional form of $q_i$.
Without loss of generality, suppose the observations are ordered in an increasing way: $s_{i,1} \le s_{i,2} \le \dots \le s_{i,m_i}$.
Then, for each given $s \in [s_{i,l}, s_{i,l + 1}]$, we estimate $q_i(s)$ by 
\begin{align}\label{eq:interp}
    \hat q_i(s) = \omega_i(s) q_i(s_{i,l}) + (1 - \omega_i(s)) q_i(s_{i,l + 1}),
\end{align}
where $\omega_i(s) = (s_{i,l + 1} - s)/(s_{i,l + 1} - s_{i,l})$.
When $s < s_{i,1}$ (resp. $s > s_{i,m_i}$), we can set $\hat q_i(s) = q(s_{i,1})$ (resp. $\hat q_i(s) = q(s_{i,m_i})$).

When $q_i$ is a quantile function, it is also typical that a finite sample $\{y_{i,1}, \ldots, y_{i,m_i}\}$ randomly drawn from $F_i$ is only available.
In this case, a straightforward approach to estimate $q_i$ would be to perform a nonparametric kernel CDF estimation and invert the estimate. 
Alternatively, we can also use a simple interpolation method as described in \cite{yang2020random}.

Letting $\hat q_i$ be any estimator of $q_i$, compute $\hat r_{i,k} \coloneqq \sum_{j = 1}^n w_{i,j} \int_0^1 \hat q_j(t) \phi_k(t) \text{d}t$ and let $\hat r_i = (\hat r_{i,1}, \ldots, \hat r_{i,K})^\top$.
Now, the feasible version of \eqref{eq:2SLS} is defined as
\begin{align}\label{eq:f2SLS}
\begin{split}
    \tilde \beta_n(s)
    & \coloneqq \left[ X^\top (I_n - \hat S) X \right]^{-1} X^\top (I_n - \hat S) \hat Q(s)\\
    \tilde \theta_n(s)
    & \coloneqq \left[ \hat R_x^\top M_z \hat R_x + \lambda D n \right]^{-1} \hat R_x^\top M_z \hat Q(s) 
\end{split}
\end{align}
where $\hat Q(s) = (\hat q_1(s), \ldots, \hat q_n(s))^\top$, $\hat R = (\hat r_1, \ldots, \hat r_n)^\top$, $\hat R_x = (I_n - M_x)\hat R$, and $\hat S = M_z \hat R [ \hat R^\top M_z \hat R]^{-} \hat R^\top M_z$.
The estimator for $\alpha_0(\cdot, s)$ can be obtained by $\tilde \alpha_n(\cdot, s) \coloneqq \sum_{k=1}^K \phi_k(\cdot) \tilde \theta_{nk}(s)$.

\subsection{Convergence rates and limiting distributions}

To derive the asymptotic properties of our estimators, we first need to specify the structure of our sampling space.
Following \cite{jenish2012spatial}, let $\mathcal{D} \subset \bbR^d$, $1 \le d < \infty$ be a possibly uneven lattice, and $\mathcal{D}_n \subset \mathcal{D}$ be the set of observation locations, which may differ across different $n$.
For spatial data, $\mathcal{D}$ would be defined by a geographical space with $d = 2$.
Notably, $\mathcal{D}$ does not necessarily have to be exactly observable to us.
For example, $\mathcal{D}$ is possibly a complex space of general social and economic characteristics. In this case, we can consider it to be an embedding of individuals in a latent space, instead of their physical locations.

\begin{assumption}\label{as:sample_space}
    (i) The maximum coordinate difference between any two observations $i,j \in \mathcal{D}$, which we denote as $\Delta(i,j)$, is at least (without loss of generality) 1; and
    (ii) a threshold distance $\bar \Delta$ exists such that $w_{i,j} = 0$ if $\Delta(i,j) > \bar \Delta$.
\end{assumption}
Assumptions \ref{as:sample_space}(i) and (ii) together imply that the number of interacting neighbors for each unit is bounded.
We believe this is not too restrictive in practice.

\begin{assumption}\label{as:covariates}
    (i) $\{z_i\}_{i = 1}^n$ are non-stochastic and uniformly bounded; and
    (ii) $\lim_{n \to \infty}Z^\top Z/n$ exists and is nonsingular.
\end{assumption}

\begin{assumption}\label{as:error}
    (i) For all $i$, $\varepsilon_i \in L^p(0,1)$ for some $2 \le p < \infty$;
    (ii) $\{\varepsilon_i\}_{i = 1}^n$ are independent; and
    (iii) $\bbE [\varepsilon_i(s)] = 0$ for all $i$, $\inf_{1 \le i \le n; \; n \ge 1}||\varepsilon_i(s)||_2 > 0$, and $\sup_{1 \le i \le n; \; n \ge 1}||\varepsilon_i(s)||_4 \lesssim 1$.  
\end{assumption}

Assumption \ref{as:covariates}(i) states that the covariates and instruments are constant.
The same type of assumption as this has been often utilized in the literatures on spatial econometrics and many-IV estimation (e.g., \citealp{kelejian2010specification, hausman2012instrumental}).
Note that this assumption is essentially equivalent to considering all stochastic arguments as being conditional on $\{z_i\}_{i=1}^n$.  
Assumption \ref{as:error} restricts the distribution of the error functions, which accommodates virtually any form of heteroscedasticity.
We might be able to relax the independence assumption in (ii) to some weak dependence condition, but we introduce this for technical simplicity.
The $s$ in (iii) is a given interior point of $[0,1]$ at which the estimation is performed.

\begin{assumption}\label{as:base}
    (i) For all $k$, $\phi_k \in L^2(0,1)$;
    (ii) $|| \alpha_0(\cdot, s) -  \bm{\phi}_K(\cdot)^\top \theta_0(s)  ||_{L^2} \le \ell_K(s)$, where $\bm{\phi}_K = (\phi_1, \ldots, \phi_K)^\top$; and
    (iii) $\rho_{\max}( \int_0^1 \bm{\phi}_K(t) \bm{\phi}_K(t)^\top \text{d}t ) \lesssim 1$.
\end{assumption}

Assumption \ref{as:base} imposes a set of conditions on the basis functions.
The $L^2$-convergence rate of the approximation errors for various bases is discussed in \cite{belloni2015some}, where it is shown that $\ell_K(s) \lesssim K^{-\pi}$ typically holds when $\alpha_0(\cdot, s)$ is a \textit{$\pi$-smooth} function (i.e., H\"older class of smoothness order $\pi$).

\begin{assumption}\label{as:illposedness} 
    (i) $\rho_{\max} ( \bbE[\bar R^\top Z/n] \bbE[Z^\top \bar R /n] ), \rho_{\max}( \bbE[\bar R_x^\top Z/n] \bbE[Z^\top \bar R_x /n] ) \lesssim 1$; and
    (ii) there exists $\nu_{KL} > 0$ such that $\nu_{KL} \le \liminf_{n \to \infty} \rho_{\min} ( \bbE[\bar R^\top Z/n] \bbE[Z^\top \bar R /n] ), \liminf_{n \to \infty} \rho_{\min}( \bbE[\bar R_x^\top Z/n] \bbE[Z^\top \bar R_x /n] )$.
\end{assumption}

\begin{remark}[Potentially weak identification of $\alpha_0$]\label{rem:ID}
    The $\nu_{KL}$ in Assumption \ref{as:illposedness}(ii) governs the strength of the identification of $\alpha_0$, conceptually equivalent to the issue of \textit{ill-posedness estimation} in high-dimensional IV regression models (\citealp{breunig2020ill}).
    It is important to note that, the ill-posedness problem in our context is a more practical concern, unlike the intrinsically ill-posed nature of nonparametric IV models (e.g., \citealp{blundell2007semi,hoshino2022sieve}).
    As mentioned in Remark \ref{rem:IV}, our model assumes only a finite number of exogenous variables (i.e., $x_i$), while the number of endogenous variables grows to infinity.
    One potential strategy for constructing a sufficient number of IVs is to use higher-order spatial lags of $x_i$.
    However, as the order of spatial lags increases, their correlation with the endogenous variables inevitably gets weaker, and the IVs themselves typically become more collinear.
    This results in the ill-posedness problem, slowing down the rate of convergence, and inflating the variance of our estimator.
    A similar discussion can be found in \cite{tchuente2019weak}.
    We introduce the penalty term $\lambda D$ to control the variance inflation by restricting the flexibility of the estimated function.
\end{remark}

To state the next assumption, we define the following matrices:
$\bm{V}_n(s) \coloneqq \text{diag}\{\bbE[\varepsilon_1^2(s)], \ldots, \bbE[ \varepsilon_n^2(s)]\}$, $\Omega_{n,x}(s) \coloneqq \Psi_{n,x}^\top \bm{V}_n(s) \Psi_{n,x}/n$,
\begin{align}
    \Psi_{n,x} 
    & \coloneqq X - Z (Z^\top Z / n)^{-} \bbE (Z^\top \bar R / n)  \left[ \bbE \bar R^\top M_z \bbE \bar R / n \right]^{-1} \bbE (\bar R^\top X / n), \\
    \Sigma_{n, x} 
    & \coloneqq X^\top X /n - \bbE(X^\top \bar R/n) \left[ \bbE \bar R^\top M_z \bbE \bar R / n \right]^{-1} \bbE(\bar R^\top X/n).
\end{align}

\begin{assumption}\label{as:covmat}
    $\Sigma_x \coloneqq \lim_{n \to \infty} \Sigma_{n,x}$ and $\Omega_x(s) \coloneqq \lim_{n \to \infty} \Omega_{n,x}(s)$ exist and are nonsingular.
\end{assumption}

The next theorem gives the convergence rate of our estimator.
\begin{theorem}\label{thm:conv}
    Suppose Assumptions \ref{as:inverse} and \ref{as:sample_space} -- \ref{as:covmat} hold. 
    In addition, assume $L\sqrt{K} / (\nu_{KL}^2 \sqrt{n}) \lesssim 1$.
    Then, we have
    \begin{align}
        \text{(i)} \;\;
        || \hat \beta_n(s) - \beta_0(s) || \lesssim_p n^{-1/2}, \; \text{and} \;
        \text{(ii)} \;\;
        & \left\| \hat \alpha_n(\cdot, s) - \alpha_0(\cdot, s) \right\|_{L^2} \lesssim_p \frac{\sqrt{K}/\sqrt{n} + \ell_K(s)}{\sqrt{\nu_{KL} + \lambda \rho_D}} + \frac{\lambda ||\theta_0(s)||_D}{\nu_{KL} + \lambda \rho_D},
    \end{align}
    where $\rho_D \coloneqq \rho_{\min}(D)$, and $||\theta_0(s)||_D \coloneqq \sqrt{\theta_0(s)^\top D \theta_0(s)}$.
\end{theorem}

The proofs of Theorem \ref{thm:conv} and those presented below are somewhat similar in several parts to those in \cite{hoshino2022sieve}, but for completeness, they are all presented in Appendix \ref{app:proof}.
Theorem \ref{thm:conv}(i) shows that the coefficients of $x_i$ can be estimated at the root-n rate.
Meanwhile, result (ii) indicates that the $L^2$-convergence rate of $\hat \alpha_n(\cdot, s)$ is not standard owing to the potential weak identification and the presence of the penalty term $\lambda D$.
We can observe a trade-off that the first term converges to zero quickly by selecting a large $\lambda$, while the second term can vanish if we select $\lambda$ diminishing at a sufficiently fast rate such that $\nu_{KL}/\lambda \to \infty$.
It is clear that the order of $||\theta_0(s)||_D$ is bounded by $\sqrt{K}$.
When $\theta_0(s)$ is a sparse vector or it is decaying in the order of basis expansion, $||\theta_0(s)||_D \lesssim 1$ might be possible.

\bigskip

Next, define $\sigma_{n, \lambda}(t, s) \coloneqq \sqrt{\bm{\phi}_K(t)^\top \Sigma_{n, r, \lambda}^{-1} \Omega_{n,r}(s) \Sigma_{n, r, \lambda}^{-1}\bm{\phi}_K(t)}$, $\Sigma_{n, r, \lambda} \coloneqq \bbE \bar R_x^\top M_z \bbE \bar R_x / n + \lambda D$, and $\Omega_{n,r}(s) \coloneqq \bbE \bar R_x^\top M_z \bm{V}_n(s) M_z \bbE \bar R_x /n$.
Moreover, let
\begin{align}
    \hat{\bm{C}}_n(s) 
    & \coloneqq \left[ X^\top (I_n - S) X /n \right]^{-1} \left( X^\top (I_n - S)\hat{\bm{V}}_n(s) (I_n - S) X /n \right) \left[ X^\top (I_n - S) X /n \right]^{-1} \\
    [\hat \sigma_{n, \lambda}(t, s)]^2 
    & \coloneqq \bm{\phi}_K(t)^\top \left[ \bar R_x^\top M_z \bar R_x/n + \lambda D \right]^{-1} \left( \bar R_x^\top M_z \hat{\bm{V}}_n(s) M_z \bar R_x /n \right) \left[ \bar R_x^\top M_z \bar R_x/n + \lambda D \right]^{-1} \bm{\phi}_K(t)
\end{align}
where $\hat{\bm{V}}_n(s) \coloneqq \text{diag}\{\hat \varepsilon_1^2(s), \ldots , \hat \varepsilon_n^2(s) \}$, and $\hat \varepsilon_i(s) \coloneqq q_i(s) - \bar r_i^\top \check \theta_n(s) - x_i^\top \hat \beta_n(s)$, where $\check \theta_n(s)$ denotes the estimator of $\theta_0(s)$ obtained following \eqref{eq:2SLS} with $\lambda$ set to zero.
Then, the limiting distribution of our estimator can be characterized as in the following theorem.

\begin{theorem}\label{thm:normality}
    Suppose Assumptions \ref{as:inverse} and \ref{as:sample_space} -- \ref{as:covmat} hold. 
    In addition, assume
    \begin{align}
        & K \sim L, \quad
        K^3/(\nu_{KL}^4 n) \to 0, \quad
        \sqrt{n} \ell_K(s)/\sqrt{\nu_{KL}} \to 0, \\
        & \sqrt{n}|\bm{\phi}_K(t)^\top \theta_0(s) - \alpha_0(t,s)| / ||\bm{\phi}_K(t)|| \to 0, \quad
        \lambda/\nu_{KL}^2 \to 0, \quad
        \sqrt{n} \lambda ||\theta_0(s)||_D / \nu_{KL}  \to 0.
    \end{align}
    Then, we have 
    \begin{align}
        \text{(i)} \;\;
        \sqrt{n} ( \hat \beta_n(s) - \beta_0(s) ) \overset{d}{\to} \mathcal{N}(0, \Sigma_x^{-1} \Omega_x(s) \Sigma_x^{-1}), \;\;
        \text{(ii)} \;\; \frac{\sqrt{n} (\hat \alpha_n(t, s) - \alpha_0(t, s))}{\sigma_{n, \lambda}(t, s)} \overset{d}{\to} \mathcal{N}(0, 1),
    \end{align}
   (iii) $\left\|\hat{\bm{C}}_n(s) - \Sigma_x^{-1} \Omega_x(s) \Sigma_x^{-1} \right\| = o_P(1)$, and (iv) $|\hat \sigma_{n, \lambda}(t, s) - \sigma_{n, \lambda}(t, s)| = o_P(1)$.
\end{theorem}

\begin{remark}[Choice of tuning parameters]
    To implement our estimator, we need to select three tuning parameters $\lambda$, $K$, and $L$.
    For the penalty parameter $\lambda$, considering the assumptions in Theorem \ref{thm:normality}, it must converge to zero faster at least than $n^{-1/2}$. 
    In the numerical studies presented below, we set $\lambda \sim n^{-3/5}$.
    For the order of basis expansion $K$, assume that $L \sim K \sim n^{\bar k}$ for some $\bar k > 0$.
    We further assume that the ill-posedness is mild such that $\nu_{KL} \lesssim K^{-\nu}$ for some $\nu > 0$ and  suppose that $\alpha_0(\cdot, s)$ is a $\pi$-smooth function such that $\ell_K(s) \lesssim K^{-\pi}$.
    Then, easy calculations yield that $K$ must satisfy $1/(2\pi - \nu) < \bar k < 1/(3 + 4\nu)$ to ensure the asymptotic normality results.
    This clearly indicates that when the IVs are not strong, a modest $K$ should be employed.
    In Section \ref{sec:MC}, we numerically examine the impact of tuning parameters selection.
    The results demonstrate that the choice of $\lambda$ is more influential on the estimation performance than that of $K$.
    More sophisticated, data-driven tuning parameter choice methods will be investigated in future studies.
\end{remark}

\subsection{Testing the presence of spatial effects}

In this section, we consider statistically testing the presence of spatial effects.
Specifically, for each given $s$, we test the following null hypothesis:
\begin{align}
    \mathbb{H}_0: \alpha_0(t, s) = 0 \;\; \text{almost everywhere $t \in \mathcal{I}$}
\end{align}
where $\mathcal{I}$ denotes a non-degenerate sub-interval of $[0,1]$.
Then, a natural test statistic for testing $\mathbb{H}_0$ would be the Wald-type statistic given as follows:
\begin{align}
    T_n \coloneqq n \int_\mathcal{I} \hat \alpha^2_n(t,s) \text{d}t,
\end{align}
where the dependence of $T_n$ on $s$ is suppressed.
To derive the asymptotic distribution of $T_n$ under $\mathbb{H}_0$, let $\Xi_n \coloneqq \Sigma_{n, r, \lambda}^{-1} \bbE (\bar R_x^\top Z / n) (Z^\top Z / n)^{-}$ and $\Phi_\mathcal{I} \coloneqq \int_\mathcal{I} \bm{\phi}_K(t) \bm{\phi}_K(t)^\top \text{d}t$.
Further, define
\begin{align}
    \mu_n 
    & \coloneqq  \text{tr}\left\{ \Xi_n^\top \Phi_\mathcal{I} \Xi_n (Z^\top \bm{V}_n(s) Z /n )\right\} \\
    v_n 
    & \coloneqq 2 \text{tr}\left\{ \Xi_n^\top \Phi_\mathcal{I} \Xi_n (Z^\top \bm{V}_n(s) Z /n ) \Xi_n^\top \Phi_\mathcal{I} \Xi_n (Z^\top \bm{V}_n(s) Z /n ) \right\},
\end{align}
which serve as the mean and variance of $T_n$, respectively.

Here, we introduce the following miscellaneous assumptions.
\begin{assumption}\label{as:misc}
    (i) $\sup_{1 \le i \le n; \; n \ge 1}||\varepsilon_i(s)||_6 \lesssim 1$; and
    (ii) $0 < \rho_\text{min}(\Phi_\mathcal{I}) \le \rho_\text{max}(\Phi_\mathcal{I}) \lesssim 1$.
\end{assumption}

The next theorem characterizes the asymptotic distribution of our test statistic.
\begin{theorem}\label{thm:test}
    Suppose Assumption \ref{as:misc} and the assumptions in Theorem \ref{thm:normality} are all satisfied.
    In addition, assume $1/(K \nu^2_{KL}) \to 0$ and $K^3/(\nu_{KL}^5 n) \to 0$.
    Then, we have $(T_n - \mu_n)/\sqrt{v_n} \overset{d}{\to} \mathcal{N}(0,1)$.
\end{theorem}

When $\mathbb{H}_0$ does not hold, the standardized test statistic $(T_n - \mu_n)/\sqrt{v_n}$ deviates to a positive value.
Thus, considering Theorem \ref{thm:test}, we can reject $\mathbb{H}_0$ at the $100\alpha$\% significance level if the realized value of $(T_n - \mu_n)/\sqrt{v_n}$ exceeds the upper $\alpha$-quantile of $\mathcal{N}(0, 1)$.
To implement the test in practice, we need to consistently estimate $\mu_n$ and $v_n$, which can be easily performed by the sample analogue estimators, the definitions of which should be clear from the context.
The consistency of these estimators is straightforward (refer to Lemmas \ref{lem:matLLN1} and \ref{lem:matLLN2} and Theorem \ref{thm:normality}(iii), (iv)).

\begin{remark}
    The proposed test can easily be extended to a more general null hypothesis: $\mathbb{H}_0: \alpha_0(t,s) = a(t)$ for $t \in \mathcal{I}$, where $a(\cdot)$ is any given function that is pre-specified by the researcher (or estimable with a certain convergence rate).
    The resulting test statistic would take the following form: $T_n = \int_\mathcal{I} (\hat \alpha_n(t,s) - a(t))^2 \text{d}t$, and $\mathbb{H}_0$ can be tested using the same procedure as above.
\end{remark}

Finally, it is important to notice that when $\mathbb{H}_0: \alpha_0(t,s) = 0$ is indeed true over the entire $[0,1]$, higher-order spatially-lagged covariates are not valid IVs, that is, for example, $\bar{\bar{x}}_i$ and $\bar q_i$ are not related to each other.
Thus, basically, we need to prepare a sufficient number of IVs using only $\bar x_i$ and possibly its transformations in this case.

\subsection{Asymptotic properties under interpolated outcome functions}

Finally, in this section, we examine the cases in which the outcome functions are only discretely observed, and they are linearly interpolated following \eqref{eq:interp}.
Letting $s_{i,0} = 0$ and $s_{i,m_i + 1} = 1$ for all $i$, we introduce the following assumption.
\begin{assumption}\label{as:interp}
    For all $i$, (i) there exists a positive sequence $\kappa \equiv \kappa_n$ tending to zero as $n$ increases such that $|s_{i,l+1} - s_{i,l}| \lesssim \kappa$, for all $l = 0,1, \ldots, m_i$; and (ii) there exists a constant $\xi \in (0,1]$ such that $|q_i(s_1) - q_i(s_2)| \lesssim |s_1 - s_2|^\xi$ for any $s_1, s_2 \in [0,1]$.
\end{assumption}
Assumption \ref{as:interp}(i) determines the overall precision of the linear interpolation approximation.
For simplicity of discussion, it assumes that the values of the outcome function are (quasi) uniformly observed such that the distance of any two consecutive observations is of order $\kappa$.
In addition, note that we treat each observation point as nonstochastic.
Assumption \ref{as:interp}(ii) requires that the outcome function is H\"{o}lder continuous with exponent $\xi$ for all $i$.
This assumption may be somewhat restrictive, but similar assumptions are often considered in the FDA literature (e.g., \citealp{crambes2009smoothing}).
Obviously, we need some form of continuity in order for the interpolation approximation to work.

The following theorem states that the approximation errors caused by the linear interpolation are asymptotically negligible if $\kappa^\xi$ is sufficiently small.

\begin{theorem}\label{thm:tilde}
    Suppose Assumption \ref{as:interp} and those in Theorem \ref{thm:normality} are all satisfied.
    In addition, assume $\sqrt{n} \kappa^\xi / \sqrt{\nu_{KL}} \to 0$.
    Then, $\tilde \beta_n(s)$ and $\tilde \alpha_n(\cdot, s)$ are asymptotically equivalent to $\hat \beta_n(s)$ and $\hat \alpha_n(\cdot, s)$, respectively.
\end{theorem}

Under Assumption \ref{as:interp}, the approximation error $|\hat q_i(s) - q_i(s)|$ is of order $\kappa^\xi$ uniformly in $s$.
The condition $\sqrt{n} \kappa^\xi / \sqrt{\nu_{KL}} \to 0$ states that the interpolation error should shrink to zero faster than $n^{-1/2}$, similar to the basis approximation error $\ell_K(s)$. 
From this result, it is also straightforward to observe the asymptotic equivalence between the feasible Wald test $\tilde T_n \coloneqq n \int_{\mathcal{I}} \tilde \alpha_n^2(t,s) \text{d}t$ and $T_n$ presented in the previous section.

\section{Numerical Experiments}\label{sec:MC}

\paragraph{Performance of the 2SLS estimator}

In this section, we first examine the finite sample performance of the proposed 2SLS estimator. 
We consider the following three data-generating processes (DGPs) for the Monte Carlo experiments:
\begin{align}
    q_i(s) = \int_0^1 \bar q_i(t) \alpha_0(t, s) \text{d}t + \sum_{j = 1}^7 x_{i,j} \beta_{0j}(s) + \varepsilon_i(s),
\end{align}
where
\begin{align*}
    \text{DGP 1:\;} & \alpha_0(t, s) = (t + s)/2\\
    \text{DGP 2:\;} & \alpha_0(t, s) = \text{PDF of $\mathcal{N}(t - s, 0.7^2)$} \\
    \text{DGP 3:\;} & \alpha_0(t, s) = 0.3 + 0.7t\sin(2\pi(t-s))
\end{align*}
$\beta_{0j}(s) = 1 + 1.2\log(s + 1)$ for $j = 1,2,3$, $\beta_{0j}(s) = \exp(s) - 0.4$ for $j = 4,\ldots, 7$, $x_{i,j} \overset{IID}{\sim} \mathcal{N}(0, 1)$ for all $j$, and $\varepsilon_i(s) = \varepsilon_{1,i} + \sum_{j = 1}^4 s^{j/2} \varepsilon_{2,i,j}$ with $\varepsilon_{1,i} \overset{IID}{\sim} \mathcal{N}(0, 0.3^2)$ and $\varepsilon_{2,i,j} \overset{IID}{\sim} \mathcal{N}(0, 0.6^2)$ for all $j$.
When estimating the model, an intercept term is also included.
We randomly allocate $n$ units on the lattice of $n/20 \times 40$, where we consider two sample sizes: $n \in \{400, 1600\}$.
The spatial weight matrix $W_n$ is defined according to the Rook contiguity with row normalization.
Since these three DGPs satisfy the requirements in Assumption \ref{as:inverse}, we can generate the outcome functions $Q$ using the Neumann series approximation: $Q \approx Q^{(L)} \coloneqq \sum_{\ell = 0}^L \mathcal{T}^\ell [X\beta_0 + \mathcal{E}]$, where $L$ is increased until $\max_{1 \le i \le n} |q_i^{(L)}(s) - q_i^{(L-1)}(s)| < 0.001$ is met for all $s$.
For computing the integrals over $[0, 1]$, we approximate them by finite summations over 199 grid points: 0.005, 0.010, \ldots, 0.995. 

For the choice of the basis functions $\{\phi_k\}$, we use the cubic B-splines.
We examine two values for the number of the inner knots of the B-splines: $\text{\# knots} \in \{2,3\}$, corresponding to $K = 6$ and $7$, respectively, both of which are equally spaced in $[0,1]$.
The IVs used are the first- and second-order spatial lags of $\{1, x_{i,1}, \ldots, x_{i,7}\}$.
Note that because there may exist some units that have no neighboring units, the spatial lags of $1$ are not necessarily constants.
For the penalty term $\lambda D$, we set $D = I_K$ (i.e., the ridge penalty) and attempt using four values for $\lambda = \lambda_c n^{-3/5}$ with $\lambda_c \in \{0.5, 1, 2, 3\}$.
The number of Monte Carlo repetitions for each setup is set to 1000.
Throughout, the evaluation point $s$ is fixed at $s = 0.5$.

The performance of the coefficient estimator $\hat \beta_n$ is evaluated using the average bias (BIAS) and the average root mean squared error (RMSE):
\begin{align}
    \text{BIAS:\;} \frac{1}{7}\sum_{j = 1}^7 \left[ \frac{1}{1000}\sum_{r = 1}^{1000} \hat \beta_{nj}^{(r)}(s) - \beta_{0j}(s) \right], \quad
    \text{RMSE:\;} \frac{1}{7}\sum_{j = 1}^7 \left[ \frac{1}{1000}\sum_{r = 1}^{1000} (\hat \beta_{nj}^{(r)}(s) - \beta_{0j}(s))^2 \right]^{1/2},
\end{align}
where superscript $(r)$ means that the estimate is obtained from the $r$-th replicated dataset.
Similarly, for the estimator $\hat \alpha_n$ of the spatial effect, we evaluate the performance based on the BIAS and RMSE averaged over the 19 evaluation points $\{t_1, t_2, \ldots, t_{19}\}$ equally spaced on $[0,1]$:
\begin{align}
    \text{BIAS:\;} \frac{1}{19}\sum_{j = 1}^{19} \left[ \frac{1}{1000}\sum_{r = 1}^{1000} \hat \alpha_n^{(r)}(t_j, s) - \alpha_0(t_j, s) \right], \quad
    \text{RMSE:\;} \frac{1}{19}\sum_{j = 1}^{19} \left[ \frac{1}{1000}\sum_{r = 1}^{1000} (\hat \alpha_n^{(r)}(t_j, s) - \alpha_0(t_j, s))^2 \right]^{1/2}.
\end{align}

Table \ref{tab:estimation} summarizes the simulation results.
Our main findings are as follows:
First, the results suggest that our estimator works satisfactorily well for all scenarios.
The RMSE values for estimating $\beta_0$ are approximately halved when the sample size is increased from 400 to 1600, which is consistent with our theorem.
Meanwhile, the RMSE values for estimating $\alpha_0$ do not decrease significantly even when the sample size is increased. 
This result would be owing to the increased variances caused by employing a smaller penalty parameter $\lambda$ (recall that $\lambda \sim n^{-3/5}$).
When comparing the results of the estimators with different $\lambda$ values, our results suggest that when the functional form of the spatial effect $\alpha_0$ is simple as in DGPs 1 and 2, using an estimator with a relatively large penalty is advisable in terms of RMSE.
In contrast, when the functional form of $\alpha_0$ is complex as in DGP 3, the estimator with the smallest penalty outperforms the others, which should be a reasonable result. 
It seems that the number of inner knots has only minute impacts on the estimation performance.

\begin{table}[ht]
    \caption{Estimation performance}
    \label{tab:estimation}
    \centering\small
    \begin{tabular}{lll|cccccccccc}
        \hline \hline
        &      &   & \multicolumn{2}{c}{$\beta$} & \multicolumn{2}{c}{$\alpha$ ($\lambda_c = 0.5$)} & \multicolumn{2}{c}{$\alpha$ ($\lambda_c = 1$)} & \multicolumn{2}{c}{$\alpha$ ($\lambda_c = 2$)} & \multicolumn{2}{c}{$\alpha$ ($\lambda_c = 3$)} \\
    DGP & $n$    & \# knots & BIAS   & RMSE  & BIAS               & RMSE  & BIAS             & RMSE  & BIAS             & RMSE  & BIAS             & RMSE  \\
    \hline
    1 & 400 & 2 & -0.0010  & 0.0361  & 0.0201  & 0.1059  & 0.0213  & 0.1024  & 0.0188  & 0.1009  & 0.0147  & 0.0996  \\
    &  & 3 & -0.0010  & 0.0362  & 0.0207  & 0.1186  & 0.0216  & 0.1146  & 0.0183  & 0.1124  & 0.0135  & 0.1108  \\
    & 1600 & 2 & -0.0004  & 0.0178  & 0.0151  & 0.0957  & 0.0189  & 0.0950  & 0.0209  & 0.0967  & 0.0207  & 0.0979  \\
    &  & 3 & -0.0004  & 0.0178  & 0.0160  & 0.1100  & 0.0196  & 0.1085  & 0.0213  & 0.1093  & 0.0207  & 0.1101  \\
   2 & 400 & 2 & -0.0017  & 0.0365  & -0.0010  & 0.0922  & -0.0054  & 0.0836  & -0.0124  & 0.0776  & -0.0187  & 0.0748  \\
    &  & 3 & -0.0018  & 0.0366  & -0.0010  & 0.1079  & -0.0057  & 0.0997  & -0.0133  & 0.0937  & -0.0202  & 0.0908  \\
    & 1600 & 2 & -0.0006  & 0.0179  & 0.0014  & 0.0852  & -0.0011  & 0.0814  & -0.0048  & 0.0781  & -0.0080  & 0.0764  \\
    &  & 3 & -0.0006  & 0.0179  & 0.0015  & 0.1018  & -0.0010  & 0.0981  & -0.0050  & 0.0948  & -0.0084  & 0.0931  \\
   3 & 400 & 2 & -0.0017  & 0.0394  & 0.0148  & 0.1821  & 0.0177  & 0.1944  & 0.0155  & 0.2060  & 0.0107  & 0.2111  \\
    &  & 3 & -0.0017  & 0.0394  & 0.0154  & 0.1850  & 0.0177  & 0.1991  & 0.0144  & 0.2115  & 0.0088  & 0.2165  \\
    & 1600 & 2 & -0.0006  & 0.0192  & 0.0077  & 0.1544  & 0.0133  & 0.1680  & 0.0172  & 0.1860  & 0.0175  & 0.1955  \\
    &  & 3 & -0.0006  & 0.0192  & 0.0087  & 0.1531  & 0.0140  & 0.1705  & 0.0173  & 0.1906  & 0.0172  & 0.2008  \\
        \hline \hline
    \end{tabular}
\end{table}

\paragraph{Performance of the Wald test}

Next, we assess the finite sample performance of our test for the presence of spatial effects.
In this analysis, we use the same DGP as given above to generate the data, with a slight modification on $\alpha_0$ in DGP 2. 
Specifically,
\begin{align}
    \alpha_0(t,s) = \varrho \times  \text{PDF of $\mathcal{N}(t - s, 0.7^2)$},
\end{align}
where $\varrho \in \{0, 0.1, 0.2\}$.
The null hypothesis to be tested is $\mathbb{H}_0: \alpha_0(t,0.5) = 0$ for $t \in [0.1, 0.9]$. 
Thus, $\mathbb{H}_0$ holds true when $\varrho = 0$. 

In Table \ref{tab:test}, we present the rejection frequency over 1000 Monte Carlo repetitions at the 10\%, 5\%, and 1\% significance levels.
The results for $\varrho = 0$ demonstrate that the size of our test is reasonably well-controlled, with at most 1--2\% deviation from the nominal levels for most cases.
When the spatial effect is mild in magnitude ($\varrho = 0.1$), the estimator with a smaller penalty ($\lambda_c = 0.5$) is not sufficiently powerful to detect the effect probably owing to its large estimation variance.
However, as expected, the power of the test can be significantly improved by increasing the sample size.
In the case of a stronger spatial effect ($\varrho = 0.2$), all tests exhibit nearly perfect power property for all sample sizes.

\begin{table}[ht]
    \caption{Rejection frequency}
    \label{tab:test}
    \centering\small
    \begin{tabular}{lll|ccccccccc}
        \hline \hline
      &      &        & \multicolumn{3}{c}{$\varrho = 0$} & \multicolumn{3}{c}{$\varrho = 0.1$} & \multicolumn{3}{c}{$\varrho = 0.2$}\\
    \# knots & $n$    & $\lambda_c$ & 10\%  & 5\%   & 1\%   & 10\%  & 5\%     & 1\%   & 10\%  & 5\%     & 1\%   \\
    \hline
    2 & 400 & 0.5 & 0.085  & 0.048  & 0.023  & 0.315  & 0.231  & 0.128  & 0.985  & 0.974  & 0.900  \\
    &  & 1 & 0.080  & 0.054  & 0.021  & 0.762  & 0.686  & 0.497  & 1.000  & 1.000  & 0.998  \\
    &  & 2 & 0.074  & 0.047  & 0.018  & 0.971  & 0.956  & 0.908  & 1.000  & 1.000  & 1.000  \\
    &  & 3 & 0.068  & 0.047  & 0.022  & 0.982  & 0.977  & 0.959  & 1.000  & 1.000  & 1.000  \\
    & 1600 & 0.5 & 0.081  & 0.058  & 0.030  & 0.642  & 0.474  & 0.260  & 1.000  & 1.000  & 1.000  \\
    &  & 1 & 0.084  & 0.056  & 0.029  & 0.980  & 0.951  & 0.811  & 1.000  & 1.000  & 1.000  \\
    &  & 2 & 0.087  & 0.057  & 0.029  & 1.000  & 1.000  & 1.000  & 1.000  & 1.000  & 1.000  \\
    &  & 3 & 0.086  & 0.056  & 0.027  & 1.000  & 1.000  & 1.000  & 1.000  & 1.000  & 1.000  \\
    3 & 400 & 0.5 & 0.085  & 0.048  & 0.023  & 0.322  & 0.236  & 0.129  & 0.986  & 0.976  & 0.912  \\
    &  & 1 & 0.080  & 0.053  & 0.021  & 0.782  & 0.695  & 0.516  & 1.000  & 1.000  & 0.998  \\
    &  & 2 & 0.077  & 0.046  & 0.017  & 0.972  & 0.957  & 0.915  & 1.000  & 1.000  & 1.000  \\
    &  & 3 & 0.069  & 0.047  & 0.023  & 0.982  & 0.979  & 0.962  & 1.000  & 1.000  & 1.000  \\
    & 1600 & 0.5 & 0.077  & 0.054  & 0.028  & 0.624  & 0.457  & 0.243  & 1.000  & 1.000  & 1.000  \\
    &  & 1 & 0.081  & 0.055  & 0.028  & 0.979  & 0.952  & 0.817  & 1.000  & 1.000  & 1.000  \\
    &  & 2 & 0.085  & 0.055  & 0.028  & 1.000  & 1.000  & 1.000  & 1.000  & 1.000  & 1.000  \\
    &  & 3 & 0.081  & 0.053  & 0.027  & 1.000  & 1.000  & 1.000  & 1.000  & 1.000  & 1.000  \\
      \hline \hline
    \end{tabular}
\end{table}

\paragraph{Simulations under discretely observed outcome functions}

Finally, we evaluate the performance of our estimator and test when the entire shapes of the outcome functions are not perfectly observed but their values are discretely observable at finite points. 
The DGPs investigated here are identical to those used previously. 
To recover the entire functional form of the outcome function for each unit, we use the linear interpolation method in \eqref{eq:interp}.
For all units, we assume that $m$ pairs of points $\{(s_{i,j}, q_i(s_{i,j}))\}_{j=1}^m$ are observable, where $s_{i,j}$'s are uniformly randomly drawn from $[0,1]$, and $m$ is selected from two values $m \in \{15, 50\}$.

To save space, the simulation results are omitted here and provided in Tables \ref{tab:estimation_interp} and \ref{tab:test_interp} in Appendix \ref{sec:app_MC}.
From these tables, we can observe similar overall tendencies as those shown above.
An interesting finding is that, although increasing $m$ from 15 to 50 improves the RMSE for most cases, there are some situations in which the estimator with a smaller $m$ achieves an even slightly better RMSE.
Similarly, comparing the results when $m = 15$ with those when the outcome function is fully observable (those reported in Table \ref{tab:estimation}), the former occasionally exhibits smaller RMSE values.
We conjecture that these phenomena occurred because the linear interpolation ``smoothed out'' the original, potentially noisier, outcome function, leading to a reduction in estimation variance.
A similar discussion can be found in \cite{imaizumi2018pca} in a different but related context.
In contrast, regarding the size property of the Wald test, the linear interpolation seems to introduce certain distortions.
Unsurprisingly, these distortions can be somewhat mitigated if $m$ is large.
Except when $\lambda_c = 0.5$, the test exhibits a satisfactory power for both values of $m$.

\section{An Empirical Illustration: Age Distribution of Japanese Cities}\label{sec:empir}

In this section, we apply the proposed estimator and test to analyze the determinants of the age distribution of Japanese cities.
While this type of data has been regularly studied in the FDA literature (e.g., \citealp{delicado2011dimensionality, hron2016simplicial, Bigot2017geodesic}), there are few papers attempting a regression-based analysis.
In recent decades, many rural Japanese cities have been facing a serious aging population, prompting them to plan campaigns to encourage young people from urban areas to settle in their cities.
Thus, investigating the relationship between the regional socioeconomic characteristics and the age structure and the impact of neighborhood trend on it would be meaningful.

Our sample comprises all local municipalities (\textit{Shi-ku-cho-son}) in Japan.
The age distribution data for each city are taken from the 2020 Census.
For the covariates to explain the age distribution, we use the ratio of agricultural, forestry, and fishery workers, number of hospital beds per capita, number of childcare facilities per capita, unemployment rate, logarithm of annual commercial sales, and logarithm of average residential landprice.
All variables are as of the most recent year before 2020, and they are all publicly available.\footnote{
    Landprice data: \url{https://www.lic.or.jp/landinfo/research.html};
    all others: \url{https://www.e-stat.go.jp/en}.
    }
In addition to these, we include five regional dummies.\footnote{
    They correspond to each of the following: \textit{Hokkaido-Tohoku}, \textit{Chubu}, \textit{Kinki}, \textit{Chugoku-Shikoku}, and \textit{Kyushu-Okinawa} regions.
}
After excluding the observations with missing items, the analysis is performed on 1883 municipalities.
Table \ref{tab:desctab} in Appendix \ref{sec:app_empir} summarizes the detailed definitions of the variables used and their basic statistics.

Our age distribution data are not complete; we only have information on the population size at five-year intervals (0 -- 4 years old, 5 -- 9 years old, and so forth).
Therefore, when computing the quantile function for each city, we performed the linear interpolation as in \eqref{eq:interp}.
In Figure \ref{fig:Quants}, we depict the obtained quantile functions for 20 randomly selected cities from our dataset.
The figure clearly shows the existence of certain regional heterogeneity in age compositions except those close to the boundary points.

\begin{figure}[ht!]
    \centering
    \includegraphics[width = 13cm]{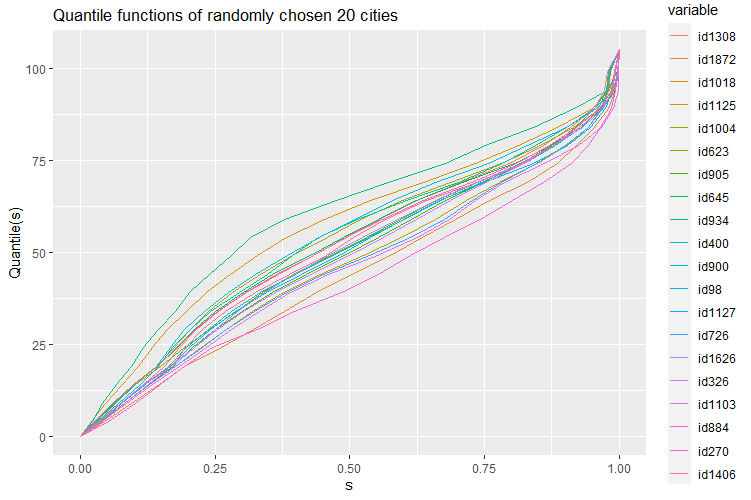}
    \caption{Age quantile functions of randomly selected cities}
    \label{fig:Quants}
\end{figure}

For estimation, we follow the same procedure as in the previous section with $K = 7$ (three inner knots) and $\lambda = 3n^{-3/5}$.
The integrals are replaced by summations over 399 equally-spaced grid points on $[0,1]$.
For the spatial weight, expecting that the impacts from demographic changes in large cities should be larger than those from small cities, we consider the following specification:
\begin{align}
    w_{i,j} = \frac{\bm{1}\{\text{$i$ and $j$ are adjacent}\}\sqrt{\text{Population}_j}}{\sum_{j \neq i}\bm{1}\{\text{$i$ and $j$ are adjacent}\}\sqrt{\text{Population}_j}}.
\end{align}
When city $i$ has no neighbors (e.g., islands), we set $w_{i,j} = 0$ for all $j$.
The estimation is performed on nine quantile values: $s = 0.1$, $0.2$, \ldots, $0.9$.

To save space, the estimated coefficients $\beta_0(s)$ are presented in Figure \ref{fig:beta} in Appendix \ref{sec:app_empir}.
Our major findings from the figure are as follows:
Interestingly, for all variables, the impacts on age distribution become prominent around the median ($s = 0.5$), suggesting the residential flexibility of this age group in response to the socioeconomic conditions of a city.
The variables considered as indicators of urbanness, such as the commercial sales and the landprice, exhibit negative effects, contributing to population rejuvenation.
As expected, cities with a higher rate of agricultural workers exhibit a significant aging trend.
Both the number of hospital beds and childcare facilities positively affect age distribution, although the underlying mechanisms are unclear.
It is important to recall that in this study, the covariates are treated as fixed, and their potential endogeneity is ignored.
To interpret the obtained results as a causal relationship, addressing the endogeneity issue more carefully would be necessary. 

The estimated spatial effect function is reported in Figure \ref{fig:alpha}.
The figure includes nine panels, each corresponding to different $s$-values.
In the figure, we also report the computed test statistic $(T_n - \mu_n)/\sqrt{v_n}$ for $\mathcal{I} = [0, 1]$.
From these results, we can observe the following:
First, the values of the test statistic suggest that the spatial effects exist significantly at all nine quantiles.
However, when quantile $t$ of the neighbor is close to either of the boundary points 0 or 1, almost no or weak spatial effects are present.
This seems reasonable considering Figure \ref{fig:Quants}; only a little regional heterogeneity in age distribution is present at these extreme quantiles.
The spatial interaction effects become particularly strong when both $t$ and $s$ are approximately 0.2 -- 0.5, which roughly correspond to the ages of the younger working population.
This result might suggest that the growth of economic activities and their spillovers play main roles in forming the spatial trend of age distribution.
Notably, the impacts from these lower-to-middle quantile values somewhat persist even for higher quantile ages.
This could be reflecting the indirect effects from positive interactions among younger age groups, rather than a direct causal relationship across different quantiles.

\begin{figure}[ht]
    \begin{minipage}[b]{0.48\linewidth}
    \centering
    \includegraphics[width=\textwidth]{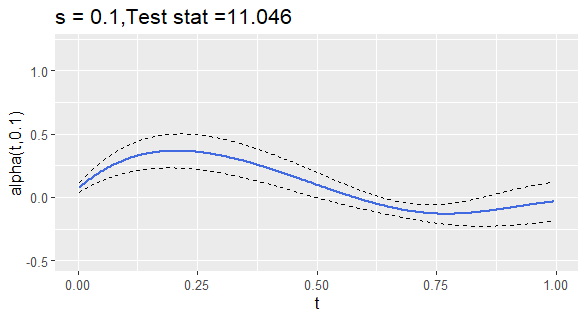}
    \end{minipage}
    \begin{minipage}[b]{0.48\linewidth}
    \centering
    \includegraphics[width=\textwidth]{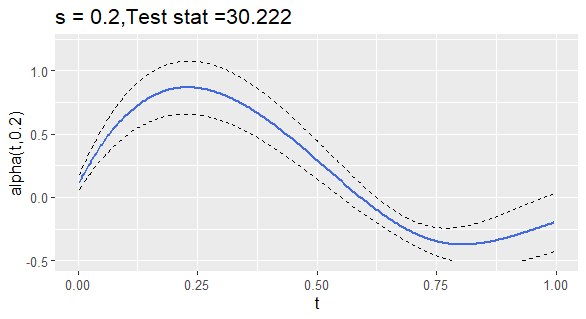}
    \end{minipage}

    \begin{minipage}[b]{0.48\linewidth}
    \centering
    \includegraphics[width=\textwidth]{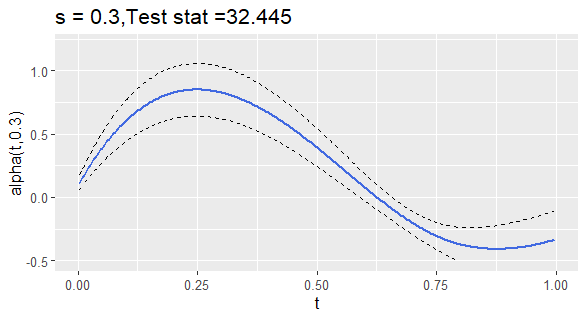}
    \end{minipage}
    \begin{minipage}[b]{0.48\linewidth}
    \centering
    \includegraphics[width=\textwidth]{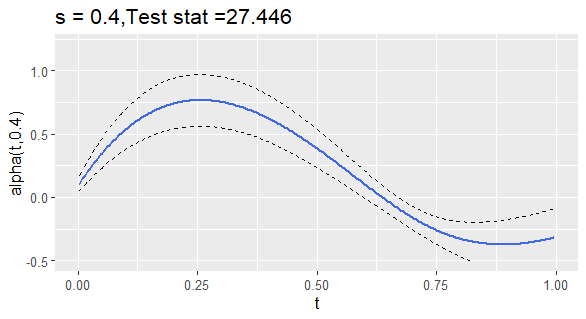}
    \end{minipage}

    \begin{minipage}[b]{0.48\linewidth}
    \centering
    \includegraphics[width=\textwidth]{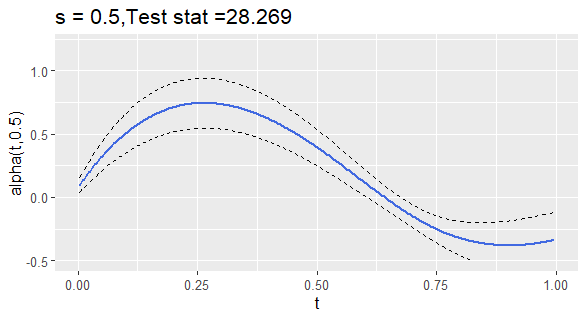}
    \end{minipage}
    \begin{minipage}[b]{0.48\linewidth}
    \centering
    \includegraphics[width=\textwidth]{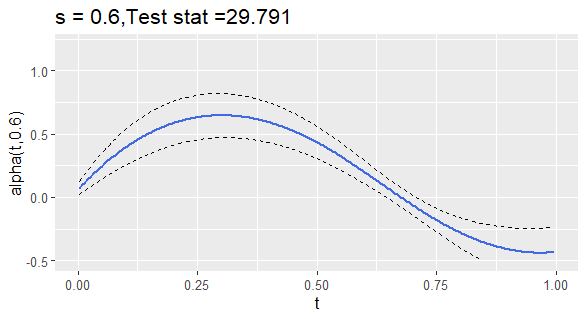}
    \end{minipage}

    \begin{minipage}[b]{0.48\linewidth}
    \centering
    \includegraphics[width=\textwidth]{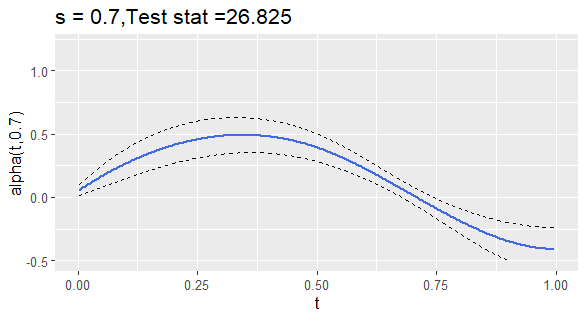}
    \end{minipage}
    \begin{minipage}[b]{0.48\linewidth}
    \centering
    \includegraphics[width=\textwidth]{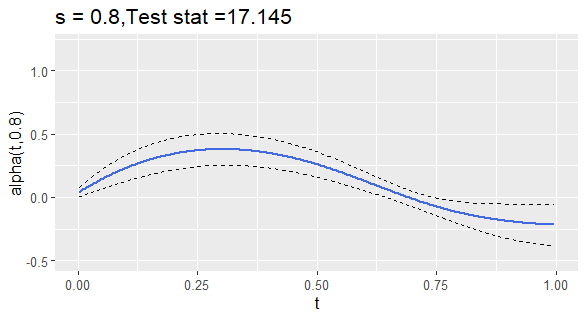}
    \end{minipage}

    \begin{minipage}[b]{0.48\linewidth}
    \centering
    \includegraphics[width=\textwidth]{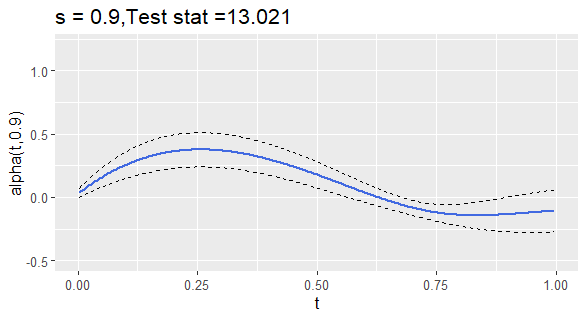}
    \end{minipage}

    \caption{Estimated spatial effect function}\label{fig:alpha}
    \centering
    \footnotesize
    In each panel, the solid line indicates the estimated $\alpha_0(\cdot, s)$, and the dotted lines indicate the 95\% confidence interval.
\end{figure}

\section{Conclusion}\label{sec:conclude}

In this study, we developed a new SAR model for analyzing spatial interactions among functional outcomes.
For estimation, we developed a penalized 2SLS estimator and established its asymptotic properties under certain regularity conditions.
Additionally, we developed a method for statistically testing the presence of spatial interactions.
To illustrate the effectiveness of our proposed method, we performed an empirical analysis focusing on the age distribution in Japanese cities.

An important potential limitation of our study is that, while we have treated the covariates as fixed variables to simplify the theoretical exposition, this approach essentially obscures the endogeneity issue underlying the covariates.
For instance, in our empirical analysis, it might be reasonable to consider the unemployment rate as an endogenous variable correlated with unobserved regional factors affecting the age distribution as well.
One way to mitigate the endogeneity issue would be to extend the current model to a panel data model with functional fixed effects, which should be a promising topic for future studies.
Another important future work is how to perform the estimation and inference when $\alpha_0$ and $\beta_0$ are not significant, leading to a weak IV problem.
We conjecture that the inclusion of additional moment conditions based on the distribution of the error term might be effective in addressing this issue, as in \cite{lee2007gmm}.
Several other issues that need future investigation include: data-driven selection of tuning parameters and developing methods for uniform inference on the functional parameters.

\clearpage

\bibliography{references.bib}

\clearpage


\appendix
\begin{center}
\Huge Appendix
\end{center}

\renewcommand{\thefigure}{\Alph{section}\arabic{figure}}
\setcounter{figure}{0}
\renewcommand{\thetable}{\Alph{section}\arabic{table}}
\setcounter{table}{0}
\renewcommand{\thefootnote}{\roman{footnote}}
\setcounter{footnote}{0}

\section{Proofs}\label{app:proof}

\begin{definition}\label{def:NED}
    Let $\bm{x} =\{ x_{n,i}: i \in \mathcal{D}_n; \; n \geq 1\}$ and $\bm{e} = \{e_{n,i} : i \in \mathcal{D}_n; \; n \geq 1\}$ be triangular arrays of random fields, where $x$ and $e$ are real-valued and general (possibly infinite-dimensional) random variables, respectively. 
    Then, the random field $\bm{x}$ is said to be $L^p$-near-epoch dependent (NED) on $\bm{e}$ if
    \begin{equation*}
    \left\| x_{n,i} - \bbE \left[ x_{n,i} \mid \mathcal{F}_{n,i}(\delta)\right] \right\|_{p} \leq c_{n,i} \varphi(\delta)
    \end{equation*}
    for an array of finite positive constants $\{c_{n,i} : i \in \mathcal{D}_n; \; n \geq 1\}$ and some function $\varphi(\delta) \geq 0$ with $\varphi(\delta) \to 0$ as $\delta\to \infty $, where $\mathcal{F}_{n,i}(\delta)$ is the $\sigma$-field generated by $\{ e_{n,j} : \Delta(i,j) \leq \delta \}$. 
    The $c_{n,i}$'s and $\varphi(\delta)$ are called the NED scaling factors and NED coefficient, respectively.
    The $\bm{x}$ is said to be uniformly $L^p$-NED on $\bm{e}$ if $c_{n,i}$ is uniformly bounded.
    If $\varphi(\delta) \lesssim \varrho^{\delta}$ for some $ 0 < \varrho <1$, then it is called geometrically $L^p$-NED. 
\end{definition}


\begin{lemma}\label{lem:NED1}
    Suppose that Assumptions \ref{as:inverse}, \ref{as:sample_space}, \ref{as:covariates}(i), and \ref{as:error}(i) hold.
    Then, for a given $s \in (0,1)$, $\{q_i(s) : i \in \mathcal{D}_n; \; n \geq 1\}$ is uniformly and geometrically $L^p$-NED on $\{\varepsilon_i: i \in \mathcal{D}_n; \; n \geq 1\}$.
\end{lemma}

\begin{proof}
    We prove the lemma in a similar manner to \cite{jenish2012nonparametric} and \cite{hoshino2022sieve}.
    First, note that $Q$ is uniquely determined in $\mathcal{H}_{n,\infty}$ as $Q = (\text{Id} - \mathcal{T})^{-1}X\beta_0 + (\text{Id} - \mathcal{T})^{-1}\mathcal{E}$ under Assumption \ref{as:inverse}.
    We denote the $i$-th element of $(\text{Id} - \mathcal{T})^{-1}X\beta_0 + (\text{Id} - \mathcal{T})^{-1}[\cdot]$ as $f_i[\cdot]$, such that $q_i = f_i[\mathcal{E}]$ holds for each $i = 1, \ldots, n$.
    
    Define 
    \[
        \mathcal{E}_{1,i}^{(\delta)} \coloneqq \{\varepsilon_j\}_{j : \Delta(i,j) \le \delta},
        \qquad
        \mathcal{E}_{2,i}^{(\delta)} \coloneqq \{\varepsilon_j\}_{j : \Delta(i,j) > \delta}
    \]
    for some $\delta > 0$.
    Since $L^p(0,1)$ is separable for $1 \le p < \infty$, under Assumption \ref{as:error}(i), both $\mathcal{E}_{1,i}^{(\delta)}$ and $\mathcal{E}_{2,i}^{(\delta)}$ are Polish space-valued random elements in $(\mathcal{H}_{|\{j : \Delta(i,j) \le \delta\}|,p}, ||\cdot||_{\infty,p})$ and $(\mathcal{H}_{|\{j : \Delta(i,j) > \delta\}|,p}, ||\cdot||_{\infty, p})$, respectively (recall: $\mathcal{H}_{n,p} \coloneqq \{H = (h_1, \ldots, h_n) : h_i \in L^p(0,1) \; \text{for all} \; i\}$). 
    Then, by Lemma 2.11 of \cite{dudley1983invariance} (see also Lemma A.1 of \cite{jenish2012nonparametric}), a function $\chi$ exists such that $(\mathcal{E}_{1,i}^{(\delta)}, \chi(U, \mathcal{E}_{1,i}^{(\delta)}))$ has the same law as that of $(\mathcal{E}_{1,i}^{(\delta)}, \mathcal{E}_{2,i}^{(\delta)})$, which is an appropriate rearrangement of $\mathcal{E}$, where $U$ is a random variable uniformly distributed on $[0,1]$ and independent of $\mathcal{E}_{1,i}^{(\delta)}$.
    
    Now, write $f_i[\mathcal{E}_{1,i}^{(\delta)}, \mathcal{E}_{2,i}^{(\delta)}] \equiv f_i[\mathcal{E}]$, and define $q^{(\delta)}_i \coloneqq f_i[\mathcal{E}_{1,i}^{(\delta)}, \chi(U, \mathcal{E}_{1,i}^{(\delta)})] \equiv f_i[\mathcal{E}^{(\delta)}]$ with $\mathcal{E}^{(\delta)} = (\varepsilon_1^{(\delta)}, \ldots, \varepsilon_n^{(\delta)})^\top$; specifically,
    \begin{align}
        q_i^{(\delta)}(s)
        & = \left\{ (\text{Id} - \mathcal{T})^{-1}[X\beta_0 + \mathcal{E}^{(\delta)}](s)\right\}_i \\
        & = \sum_{j = 1}^n w_{i,j}\int_0^1 q_j^{(\delta)}(t) \alpha_0(t,s) \text{d}t + x_i^\top \beta_0(s) + \varepsilon_i^{(\delta)}(s).
    \end{align}
    By construction, we have 
    \begin{align*}
    \bbE [q_i(s) \mid \mathcal{F}_{n,i}(\delta)] 
    & = \bbE \left[f_i[\mathcal{E}_{1,i}^{(\delta)}, \mathcal{E}_{2,i}^{(\delta)}](s) \mid \mathcal{E}_{1,i}^{(\delta)} \right] \\
    & = \bbE \left[f_i[\mathcal{E}_{1,i}^{(\delta)}, \chi(U, \mathcal{E}_{1,i}^{(\delta)})](s) \mid \mathcal{E}_{1,i}^{(\delta)}\right]
    = \bbE [q^{(\delta)}_i(s) \mid \mathcal{F}_{n,i}(\delta)],
    \end{align*}
    where $\mathcal{F}_{n,i}(\delta)$ is the $\sigma$-field generated by $\{ \varepsilon_j : \Delta(i,j) \le \delta \}$.

    Here, suppose that $0 < \delta < \bar \Delta$, where $\bar \Delta$ is as provided in Assumption \ref{as:sample_space}(ii). 
    Then, because at least $i$'s own $\varepsilon_i$ is included in $\mathcal{E}_{1,i}^{(\delta)}$, we have  $\varepsilon_i \equiv \varepsilon_i^{(\delta)}$, and hence 
    \begin{align}
        q_i(s) - q^{(\delta)}_i(s)  = \sum_{j=1}^n w_{i,j} \int_0^1 [q_j(t) - q^{(\delta)}_j(t)]\alpha_0(t,s)\text{d}t
    \end{align}
    holds.
    By Minkowski's inequality,
    \begin{align*}
    \left\| q_i(s) - q^{(\delta)}_i(s) \right\|_p 
    & = \left\| \sum_{j=1}^n w_{i,j} \int_0^1 [q_j(t) - q^{(\delta)}_j(t)]\alpha_0(t,s)\text{d}t \right\|_p  \\
    & \le \sum_{j = 1}^n | w_{i,j} | \cdot \int_0^1 \left\| q_j(t) - q^{(\delta)}_j(t)\right\|_p | \alpha_0(t,s)| \text{d}t \\
    & \le 2 \bar \alpha_0 ||W_n||_\infty \max_{1 \le j \le n} \int_0^1 \left\| q_j(t) \right\|_p \text{d}t \le C_p \cdot \varrho,
    \end{align*}
    where $C_p \coloneqq 2 \sup_{1 \le i \le n ; \: n \ge 1} \int_0^1 || q_i(t)||_p \text{d}t$, and $\varrho \coloneqq \bar \alpha_0 ||W_n||_\infty$.
    Similarly, when $\bar \Delta \leq \delta < 2 \bar \Delta$ holds, noting now that under Assumption \ref{as:sample_space}(ii) we have $\varepsilon_j \equiv \varepsilon^{(\delta)}_j$ for all $j$'s who are direct neighbors of $i$,
    \begin{align*}
    \begin{split}
    \left\| q_i(s) - q^{(\delta)}_i(s) \right\|_p
    & \le \bar \alpha_0 \sum_{j = 1}^n | w_{i,j} | \cdot \int_0^1 \left\| q_j(t_1) - q^{(\delta)}_j(t_1)\right\|_p \text{d}t_1 \\
    & = \bar \alpha_0 \sum_{j = 1}^n | w_{i,j} | \cdot \int_0^1 \left\| \sum_{k=1}^n w_{j,k} \int_0^1 [q_k(t_2) - q^{(\delta)}_k(t_2)]\alpha_0(t_2,t_1)\text{d}t_2 \right\|_p \text{d}t_1 \\
    & = \bar \alpha_0^2 \sum_{j = 1}^n | w_{i,j} | \sum_{k = 1}^n | w_{j,k} | \cdot \int_0^1  \int_0^1 \left\|  q_k(t_2) - q^{(\delta)}_k(t_2)\right\|_p \text{d}t_2 \text{d}t_1  \le C_p \cdot \varrho^2.
    \end{split}
    \end{align*}
    Applying the same argument recursively, for $m \bar \Delta \leq \delta < (m + 1)\bar \Delta$ such that $\varepsilon_j \equiv \varepsilon^{(\delta)}_j$ for all $j$'s in the $m$-th order neighborhood of $i$, we obtain
    \begin{align}\label{eq:qdiff}
    \left\| q_i(s) - q^{(\delta)}_i(s) \right\|_p  \leq C_p \cdot \varrho^{\lfloor \delta / \bar \Delta \rfloor + 1}.
    \end{align}
    Finally, by Jensen's inequality and \eqref{eq:qdiff},
    \begin{align}\label{eq:y_NED}
    \begin{split}
    \left\| q_i(s) - \bbE [q_i(s) \mid \mathcal{F}_{n,i}(\delta)] \right\|_p
    & = \left\| \int_0^1 \left[ f_i[\mathcal{E}_{1,i}^{(\delta)}, \mathcal{E}_{2,i}^{(\delta)}](s) - f_i[\mathcal{E}_{1,i}^{(\delta)}, \chi(u, \mathcal{E}_{1,i}^{(\delta)})](s) \right] \mathrm{d}u \right\|_p \\
    & \leq \left\{ \bbE \int_0^1 \left| f_i[\mathcal{E}_{1,i}^{(\delta)}, \mathcal{E}_{2,i}^{(\delta)}](s) - f_i[\mathcal{E}_{1,i}^{(\delta)}, \chi(u, \mathcal{E}_{1,i}^{(\delta)})](s) \right|^p \mathrm{d}u \right\}^{1/p} \\
    & = \left\{ \bbE \left| f_i[\mathcal{E}_{1,i}^{(\delta)}, \mathcal{E}_{2,i}^{(\delta)}](s) - f_i[\mathcal{E}_{1,i}^{(\delta)}, \chi(U, \mathcal{E}_{1,i}^{(\delta)})](s) \right|^p \right\}^{1/p} \\
    & = \left\| f_i[\mathcal{E}_{1,i}^{(\delta)}, \mathcal{E}_{2,i}^{(\delta)}](s) - f_i[\mathcal{E}_{1,i}^{(\delta)}, \chi(U, \mathcal{E}_{1,i}^{(\delta)})](s)\right\|_p\\
    & = \left\| q_i(s) - q^{(\delta)}_i(s) \right\|_p \le C_p \cdot \varrho^{\lfloor \delta / \bar \Delta \rfloor + 1} \to 0
    \end{split}
    \end{align}
    as $\delta \to \infty$ by Assumption \ref{as:inverse}.
    This completes the proof.
\end{proof}


\begin{lemma}\label{lem:cov}
    Suppose that $\{x_{n,i}: i \in \mathcal{D}_n; \; n \geq 1\}$ is geometrically $L^2$-NED on $\{\varepsilon_i: i \in \mathcal{D}_n; \; n \geq 1\}$.
    Then, under Assumption \ref{as:error}(ii), $\left| \text{Cov}\left(x_{n,i}, x_{n,j} \right) \right| \le 32 (\max_{1 \le i \le n}||x_{n,i}||_2)^2 \varphi(\Delta(i,j) / 3)$ for some geometric NED coefficient $\varphi$.
\end{lemma}

\begin{proof}
    Decompose $x_{n,i} = x_{n,1,i}^{(\delta)} + x_{n,2,i}^{(\delta)}$, where
    \begin{align}
      x_{n,1,i}^{(\delta)} \coloneqq \bbE\left[ x_{n,i} \mid \mathcal{F}_{n,i}(\delta) \right], 
      \;\; \text{and} \;\; x_{n,2,i}^{(\delta)} \coloneqq x_{n,i} - \bbE\left[ x_{n,i} \mid \mathcal{F}_{n,i}(\delta)\right].
    \end{align}
    Then, for each pair $x_{n,i}$ and $x_{n,j}$, setting: $\delta=\Delta(i,j)/3$,
    \begin{align}
      \left| \mathrm{Cov}\left( x_{n,i},x_{n,j}\right) \right|  
      & = \left| \mathrm{Cov}\left( x_{n,1,i}^{(\Delta(i,j)/3)} + x_{n,2,i}^{(\Delta(i,j)/3)}, x_{n,1,j}^{(\Delta(i,j)/3)} + x_{n,2,j}^{(\Delta(i,j)/3)} \right) \right| \\
      & \le \left| \mathrm{Cov}\left( x_{n,1,i}^{(\Delta(i,j)/3)}, x_{n,1,j}^{(\Delta(i,j)/3)}\right) \right| + \left| \mathrm{Cov}\left( x_{n,1,i}^{(\Delta(i,j)/3)},x_{n,2,j}^{(\Delta(i,j)/3)}\right) \right| \\
      & \quad + \left| \mathrm{Cov}\left( x_{n,2,i}^{(\Delta(i,j)/3)}, x_{n,1,j}^{(\Delta(i,j)/3)}\right) \right| + \left| \mathrm{Cov}\left( x_{n,2,i}^{(\Delta(i,j)/3)}, x_{n,2,j}^{(\Delta(i,j)/3)}\right) \right|.
    \end{align}
    Since $\{\varepsilon_k : \Delta(i,k) \le \Delta(i,j)/3\}$ and $\{\varepsilon_k : \Delta(j,k) \le \Delta(i,j)/3\}$ do not overlap, the first term on the right-hand side is zero by Assumption \ref{as:error}(ii). 
    Note that, by Jensen's inequality, $|| x_{n,1,i}^{(\Delta(i,j)/3)}||_2 \le || x_{n,i} ||_2$.
    In addition, $|| x_{n,2,i}^{(\Delta(i,j)/3)}||_2 \le 2 || x_{n,i} ||_2$.
    Then, since $\{x_{n,i}\}$ is assumed to be $L^2$-NED, it holds that
    \begin{align}
        \left\| x_{n,2,i}^{(\Delta(i,j)/3)} \right\|_2
        =  \left\| x_{n,i} - \bbE\left[ x_{n,i} \mid \mathcal{F}_{n,i}(\Delta(i,j)/3)\right] \right\|_2 
        \le 2 \max_{1 \le i \le n} ||x_{n,i}||_2 \varphi(\Delta(i,j) / 3).
    \end{align}
    Hence, Cauchy--Schwarz inequality gives
    \begin{align}
      \left| \mathrm{Cov}\left( x_{n,1,i}^{(\Delta(i,j)/3)}, x_{n,2,j}^{(\Delta(i,j)/3)}\right) \right| 
      & \le 4 \left\| x_{n,1,i}^{(\Delta(i,j)/3)} \right\|_2 \left\| x_{n,2,j}^{(\Delta(i,j)/3)}\right\|_2  \le 8 \left( \max_{1 \le i \le n} ||x_{n,i}||_2 \right)^2 \varphi(\Delta(i,j)/3).
    \end{align}
    Similarly, 
    \begin{align}
      \left| \mathrm{Cov}\left( x_{n,2,i}^{(\Delta(i,j)/3)}, x_{n,2,j}^{(\Delta(i,j)/3)}\right) \right| 
      & \le 4\left\| x_{n,2,i}^{(\Delta(i,j)/3)}\right\|_2\left\| x_{n,2,j}^{(\Delta(i,j)/3)}\right\|_2 \le 16 \left( \max_{1 \le i \le n} ||x_{n,i}||_2 \right)^2  \varphi(\Delta(i,j)/3).
    \end{align}
    This completes the proof.
\end{proof}


\begin{lemma}\label{lem:matLLN1}
    Suppose that Assumptions \ref{as:inverse}, \ref{as:sample_space}, \ref{as:covariates}(i), \ref{as:error}(i)--(ii), and \ref{as:base}(i) hold.
    Then,
    \begin{align}
    \text{(i)} \quad \left\|  Z^\top \bar R /n  - \bbE(Z^\top \bar R / n) \right\| \lesssim_p \frac{\sqrt{KL}}{\sqrt{n}}.
    \end{align}
    If Assumption \ref{as:interp} additionally holds, we have
    \begin{align}
        \text{(ii)} \quad \left\|  Z^\top \hat R /n  - \bbE(Z^\top \bar R / n) \right\| \lesssim_p  \frac{\sqrt{KL}}{\sqrt{n}} + \kappa^\xi \sqrt{KL}.
    \end{align}
\end{lemma}

\begin{proof}
    (i) Write $z_i = (z_{1,i}^\top, x_i^\top)^\top = (z_{i,1}, \ldots, z_{i,L+d_x})^\top$.
    By Assumption \ref{as:covariates}(i),
    \begin{align}\label{eq:Markov1}
    \begin{split}
    \bbE\left\| \frac{1}{n}\sum_{i=1}^n \left( z_i \bar r^\top_i - \bbE(z_i \bar r^\top_i) \right) \right\|^2 
    & = \frac{1}{n^2} \sum_{k = 1}^K \sum_{l = 1}^{L + d_x}	\bbE\left\{\sum_{i=1}^n \left( z_{i, l} \bar r_{i,k} - \bbE(z_{i, l} \bar r_{i,k} ) \right) \right\}^2 \\
    & = \frac{1}{n^2} \sum_{k =1}^K \sum_{l = 1}^{L + d_x} \sum_{i=1}^n \text{Var}\left(z_{i, l} \bar r_{i,k} \right) \\
    & \quad + \frac{1}{n^2} \sum_{k = 1}^K \sum_{l = 1}^{L + d_x} \sum_{i=1}^n \sum_{j \neq i}^n \text{Cov}\left(z_{i, l} \bar r_{i,k}, z_{j, l} \bar r_{j,k} \right) \\
    & \lesssim  \frac{L}{n^2} \sum_{k =1}^K \sum_{i=1}^n \bbE \left[ \bar r_{i,k}^2 \right] + \frac{L}{n^2} \sum_{k = 1}^K  \sum_{i=1}^n \sum_{j \neq i}^n \left| \text{Cov}\left(\bar r_{i,k}, \bar r_{j,k} \right) \right|.
    \end{split}
    \end{align}
    By Cauchy--Schwarz inequality, $|\int_0^1 q_i(s) \phi_k(s) \text{d}s| \le ||q_i||_{L^2} ||\phi_k||_{L^2} < \infty$ by Assumption \ref{as:base}(i).
    Thus,  
    \begin{align}\label{eq:L2_r}
        \bbE \left[ \bar r_{i,k}^2 \right] 
        & = \sum_{l=1}^n \sum_{j=1}^n w_{i,j} w_{i,l} \bbE\left[\int_0^1 q_j(s_1) \phi_k(s_1) \text{d}s_1 \int_0^1 q_l(s_2) \phi_k(s_2) \text{d}s_2 \right] < \infty
    \end{align}
    uniformly in $i$, implying that the first term on the right-hand side of \eqref{eq:Markov1} is of order $KL/n$.

    Next, by Lemma \ref{lem:NED1} as $\delta \to \infty$,
    \begin{align*}
    \left\| \bar r_{i,k} - \bbE [ \bar r_{i,k} \mid \mathcal{F}_{n,i}(\delta)] \right\|_2 
    & = \left\| \sum_{j=1}^n w_{i,j} \int_0^1 q_j(t) \phi_k(t)\text{d}t - \sum_{j=1}^n w_{i,j}  \int_0^1 \bbE \left[ q_j(t) \mid \mathcal{F}_{n,i}(\delta)\right] \phi_k(t)\text{d}t \right\|_2 \\
    & \le \sum_{j=1}^n |w_{i,j}| \left\|\int_0^1 ( q_j(t) - \bbE [ q_j(t) \mid \mathcal{F}_{n,i}(\delta)] ) \phi_k(t)\text{d}t \right\|_2 \\
    & \le \sum_{j=1}^n |w_{i,j}| \int_0^1 \left\| q_j(t) - \bbE [ q_j(t) \mid \mathcal{F}_{n,i}(\delta)] \right\|_2 |\phi_k(t)| \text{d}t \\
    & \leq ||W_n||_\infty C_2 \cdot \int_0^1 |\phi_k(t)| \text{d}t \cdot \varrho^{\lfloor \delta / \bar \Delta \rfloor + 1} \to 0.
    \end{align*}
    Thus, $\{\bar r_{i,k}\}$ is uniformly and geometrically $L^2$-NED on $\{\varepsilon_i\}$, and from Lemma \ref{lem:cov} and \eqref{eq:L2_r}, $\left| \text{Cov}\left(\bar r_{i,k}, \bar r_{j,k} \right) \right| \lesssim \varphi(\Delta(i,j)/3)$, where $\varphi(\delta)$ is some geometric NED coefficient.
    Then, Lemma A.1(iii) of \cite{jenish2009central} gives
    \begin{align}
    \begin{split}
        \frac{1}{n} \sum_{k = 1}^K  \sum_{i=1}^n \sum_{j \neq i}^n \left| \text{Cov}\left(\bar r_{i,k}, \bar r_{j,k} \right) \right|
        & \lesssim  \frac{1}{n} \sum_{k=1}^K \sum_{i=1}^n \sum_{j \neq i}^n \varphi(\Delta(i,j)/3)  \\
        & =  \frac{1}{n} \sum_{k=1}^K \sum_{i=1}^n \sum_{m = 1}^\infty \sum_{j: \; \Delta(i,j) \in [m, m+1)}  \varphi(\Delta(i,j)/3) \\
        & \lesssim \frac{1}{n} \sum_{k=1}^K \sum_{i=1}^n \sum_{m = 1}^\infty m^{d-1} \varphi(m) \lesssim K,
    \end{split}
    \end{align}
    where the last equality holds by the geometric NED property.
    
    Combining these results, by Markov's inequality, we have the desired result.

    \bigskip

    (ii) By the triangle inequality,
    \begin{align}
        \left\| Z^\top \hat R /n  - \bbE(Z^\top \bar R / n) \right\| 
        & \le \left\| Z^\top (\hat R - \bar R) / n \right\| + \underbracket{\left\| Z^\top \bar R /n  - \bbE(Z^\top \bar R / n) \right\|}_{\lesssim_p \sqrt{KL}/\sqrt{n} \; \text{by (i)}}.
    \end{align}
    For the first term on the right-hand side, observe that
    \begin{align}
        Z^\top (\hat R - \bar R) / n
        & = \frac{1}{n} \sum_{i=1}^n z_i (\hat r_{i,k} - \bar r_{i,k}) \\
        & = \frac{1}{n} \sum_{i=1}^n \sum_{j = 1}^n z_i  w_{i,j} \int_0^1 \left[ \hat q_j(t) - q_j(t) \right] \bm{\phi}_K(t)^\top \text{d}t,
    \end{align}
    and hence $\left\| Z^\top (\hat R - \bar R) / n \right\|
     \le \frac{1}{n} \sum_{i=1}^n \sum_{j = 1}^n ||z_i|| \cdot |w_{i,j}| \cdot \left\| \int_0^1 [\hat q_j(t) - q_j(t)] \bm{\phi}_K(t) \text{d}t \right\|$.
    Here, define
    \begin{align}
        \omega^*_i(t)
        = \begin{cases}
            0 & \text{if} \;\; t < s_{i,1} \\
            \frac{s_{i,l(t) + 1} - t}{s_{i,l(t) + 1} - s_{i,l(t)}} & \text{if} \;\; s_{i,1} \le t \le s_{i,m_i} \; \text{with} \; t \in [s_{i,l(t)}, s_{i,l(t) + 1}] \\
            1 & \text{if} \;\; t > s_{i,m_i} 
        \end{cases}
    \end{align}
    such that we can write for all $t \in [0,1]$
    \begin{align}
        \hat q_i(t) = \omega^*_i(t) [q_i(s_{i,l(t)}) - q_i(s_{i,l(t) + 1})] + q_i(s_{i,l(t) + 1})
    \end{align}
    (recall: $s_{i,0} = 0$ and $s_{i, m_i + 1} = 1$).
    Thus, under Assumption \ref{as:interp},
    \begin{align}\label{eq:interp_err}
        \left| \hat q_i(t) - q_i(t) \right|
        & \le \left|  \omega^*_i(t) [q_i(s_{i,l(t)}) - q_i(s_{i,l(t) + 1})]\right| + \left|q_i(s_{i,l(t) + 1}) - q_i(t) \right| \\
        & \le \left| s_{i,l(t)} - s_{i,l(t) + 1} \right|^\xi + \left|s_{i,l(t) + 1} - t \right|^\xi \lesssim \kappa^\xi,
    \end{align}
    uniformly in $t$, leading to
    \begin{align}
        \left\| \int_0^1 [\hat q_i(t) - q_i(t)] \bm{\phi}_K(t) \text{d}t \right\|
        & \le  \int_0^1 \left\| [\hat q_i(t) - q_i(t)] \bm{\phi}_K(t) \right\| \text{d}t \\
        & \lesssim \kappa^\xi \int_0^1 \left\| \bm{\phi}_K(t) \right\| \text{d}t \\         
        & \lesssim \sqrt{K} \kappa^\xi
    \end{align}
    for all $i$, where the last line follows from $\int_0^1 \left\| \bm{\phi}_K(t) \right\| \text{d}t = \int_0^1 \left( \sum_{k = 1}^K  \phi^2_k(t) \right)^{1/2} \text{d}t \le \left( \sum_{k = 1}^K || \phi_k ||_{L^2}^2 \right)^{1/2}$.
    This completes the proof.
\end{proof}
	

\begin{lemma}\label{lem:matLLN2}
    Suppose that Assumptions \ref{as:inverse}, \ref{as:sample_space}, \ref{as:covariates}, \ref{as:error}(i)--(ii), \ref{as:base}(i), and \ref{as:illposedness} hold.
    Then,
    \begin{align}
        \text{(i)} 
        & \quad \left\|  \bar R^\top M_z \bar R/ n  - \bbE \bar R^\top M_z \bbE \bar R/ n \right\| \lesssim_p \frac{\sqrt{KL}}{\sqrt{n}} \\
        \text{(ii)} 
        & \quad \left\|  \left[ \bar R^\top M_z \bar R/ n \right]^{-1} - \left[ \bbE \bar R^\top M_z \bbE \bar R/ n \right]^{-1} \right\| \lesssim_p  \frac{\sqrt{KL}}{\nu_{KL}^2 \sqrt{n}} 
    \end{align}
    If Assumption \ref{as:interp} additionally holds, we have
    \begin{align}
        \text{(iii)} 
        & \quad \left\|  \hat R^\top M_z \hat R/ n  - \bbE \bar R^\top M_z \bbE \bar R/ n \right\| \lesssim_p \frac{\sqrt{KL}}{\sqrt{n}} + \kappa^\xi \sqrt{KL} \\
        \text{(iv)} 
        & \quad \left\| \left[ \hat R^\top M_z \hat R/ n \right]^{-1} - \left[ \bbE \bar R^\top M_z \bbE \bar R/ n \right]^{-1} \right\| \lesssim_p  \frac{\sqrt{KL}}{\nu_{KL}^2 \sqrt{n}} + \frac{\kappa^\xi \sqrt{KL}}{\nu_{KL}^2}.
    \end{align}
\end{lemma}

\begin{proof}
    The proofs are analogous to the proof of Lemma A.7 in \cite{hoshino2022sieve}, and thus are omitted.
\end{proof}


\begin{flushleft}
    \textbf{Proof of Theorem \ref{thm:conv}}
\end{flushleft}

(i) Letting $U(s) = (u_1(s), \ldots, u_n(s))^\top$, we write
\begin{align}
    Q(s) = \bar R \theta_0(s) + X \beta_0(s) + \mathcal{E}(s) + U(s).
\end{align}
Observe that
\begin{align}
    u_i(s)
    & = \sum_{j=1}^n w_{i,j} \left[ \int_0^1 q_j(t) \alpha_0(t,s)\text{d}t - \sum_{k=1}^K r_{j,k}\theta_{0k}(s) \right]\\
    & = \sum_{j=1}^n w_{i,j} \left[ \int_0^1 q_j(t) \alpha_0(t,s)\text{d}t - \sum_{k=1}^K \int_0^1 q_j(t) \phi_k(t) \text{d}t \cdot \theta_{0k}(s) \right]\\
    & = \sum_{j=1}^n w_{i,j} \int_0^1 q_j(t) \left\{ \alpha_0(t,s) - \sum_{k=1}^K \phi_k(t) \theta_{0k}(s)  \right\} \text{d}t.
\end{align}
Hence, by Cauchy--Schwarz inequality,
\begin{align}\label{eq:bias}
    |u_i(s)|
    & \le \left\| W_n \right\|_\infty \cdot \max_{1 \le i \le n}||q_i||_{L^2} \cdot \left\| \alpha_0(\cdot,s) - \sum_{k=1}^K \phi_k(\cdot) \theta_{0k}(s)  \right\|_{L^2} \lesssim \ell_K(s)
\end{align}
uniformly in $i$.

Noting that $X^\top S \bar R = X^\top \bar R$, we decompose
\begin{align*}
    \hat \beta_n(s) - \beta_0(s)
    & = \left[ X^\top (I_n - S) X \right]^{-1} X^\top (I_n - S) Q(s) - \beta_0(s) \\
    & = \left[ X^\top (I_n - S) X \right]^{-1} X^\top (I_n - S) [\bar R \theta_0(s) + \mathcal{E}(s) + U(s)] \\
    & = \left[ X^\top (I_n - S) X \right]^{-1} X^\top (I_n - S) \mathcal{E}(s) + \left[ X^\top (I_n - S) X \right]^{-1} X^\top (I_n - S) U(s)\\
    & \eqqcolon A_1 + A_2, \;\; \text{say.} 
\end{align*}
By a straightforward matrix-norm calculation (see, e.g., Fact A.2 in \cite{hoshino2022sieve}) and Lemmas \ref{lem:matLLN1}(i) and \ref{lem:matLLN2}(ii)
\begin{align}
    & \left\| X^\top S X /n - \bbE(X^\top \bar R/n) \left[\bbE \bar R^\top M_z \bbE \bar R /n \right]^{-1} \bbE(\bar R^\top X/n)  \right\| \\
    & \lesssim \rho_{\max}\left( (\bar R^\top X/n)  (X^\top \bar R/n) \right) \left\| \left[\bar R^\top M_z \bar R /n \right]^{-1} - \left[\bbE \bar R^\top M_z \bbE \bar R /n \right]^{-1} \right\| \\
    & \quad + \rho_{\max}\left(  \left[\bbE \bar R^\top M_z \bbE \bar R /n \right]^{-1} \right) \left\| \bar R^\top X/n - \bbE(\bar R^\top X/n) \right\| \lesssim_p \frac{\sqrt{KL}} {\nu_{KL}^2 \sqrt{n}} + \frac{\sqrt{K}}{\nu_{KL} \sqrt{n}}.
\end{align}
This implies that
\begin{align}\label{eq:sigma_X}
\left\|  X^\top (I_n - S) X / n  - \Sigma_{n,x} \right\| \lesssim_p \frac{\sqrt{KL}} {\nu_{KL}^2 \sqrt{n}} \to 0,
\end{align}
where recall that
\begin{align}
    \Sigma_{n, x} 
    = X^\top X /n - \bbE(X^\top \bar R/n) \left[ \bbE \bar R^\top M_z \bbE \bar R / n \right]^{-1} \bbE(\bar R^\top X/n).
\end{align}
Under Assumption \ref{as:covmat}, \eqref{eq:sigma_X} ensures that $\rho_{\text{min}}( X^\top (I_n - S) X / n  ) > 0$ with probability approaching one.
Because the eigenvalue of an idempotent matrix is at most one, we obtain by \eqref{eq:bias} that
\begin{align}
    ||A_2||^2
    & =  U(s)^\top (I_n - S) X \left[ X^\top (I_n - S) X \right]^{-2} X^\top (I_n - S) U(s) \\
    & \lesssim_p U(s)^\top (I_n - S) X \left[ X^\top (I_n - S) X \right]^{-1} X^\top (I_n - S) U(s) / n  \\
    & \lesssim_p \frac{1}{n} \sum_{i = 1}^n |u_i(s)|^2 \lesssim \ell^2_K(s).
\end{align}
Hence, $ ||A_2|| \lesssim_p \ell_K(s)$.

Next, for $A_1$, decompose $A_1 \coloneqq A_{11} + A_{12}$, where
\begin{align}
    A_{11} & \coloneqq \Sigma_{n,x}^{-1} \Psi_{n,x}^\top \mathcal{E}(s)/n \\
    A_{12} & \coloneqq \left[ X^\top (I_n - S) X /n \right]^{-1} X^\top (I_n - S) \mathcal{E}(s)/n - \Sigma_{n,x}^{-1} \Psi_{n,x}^\top \mathcal{E}(s)/n
\end{align}
and recall that
\begin{align}
    \Psi_{n,x} = X - \bbE (M_z \bar R / n)  \left[ \bbE \bar R^\top M_z \bbE \bar R / n \right]^{-1} \bbE (\bar R^\top X / n).
\end{align}
By Assumptions \ref{as:error}(ii) and \ref{as:covmat}, $\bbE ||A_{11}||^2 = \text{tr}\left\{ \Sigma_{n,x}^{-1} \Omega_{n,x} \Sigma_{n,x}^{-1}\right\}/n \lesssim 1/n$.
Hence, it follows from Markov's inequality that $||A_{11}|| \lesssim_p n^{-1/2}$. 

Next, let
\begin{align}
    \mathcal{P} & \coloneqq (Z^\top Z / n)^{-} \bbE (Z^\top \bar R / n)  \left[ \bbE \bar R^\top M_z \bbE \bar R / n \right]^{-1} \bbE (\bar R^\top Z / n) (Z^\top Z / n)^{-} \\
    \hat{\mathcal{P}} & \coloneqq (Z^\top Z / n)^{-} (Z^\top \bar R / n)  \left[ \bar R^\top M_z \bar R / n \right]^{-1} (\bar R^\top Z / n) (Z^\top Z / n)^{-}
\end{align}
such that we can write $X^\top (I_n - S)\mathcal{E}(s)/n = X^\top \mathcal{E}(s) / n - (X^\top Z /n) \hat{\mathcal{P}} (Z^\top \mathcal{E}(s)/n)$ and $\Psi_{n,x}^\top \mathcal{E}(s) / n = X^\top \mathcal{E}(s) / n - (X^\top Z/n) \mathcal{P} (Z^\top \mathcal{E}(s) / n)$.
$A_{12}$ can be decomposed further into three terms in the following manner:
\begin{align}
    A_{12}
    & = \left\{ \left[ X^\top (I_n - S) X /n \right]^{-1} -  \Sigma_{n,x}^{-1} \right\} X^\top \mathcal{E}(s)/n \\
    & \quad - \left\{ \left[ X^\top (I_n - S) X /n \right]^{-1} -  \Sigma_{n,x}^{-1} \right\} (X^\top Z /n) \hat{\mathcal{P}} (Z^\top \mathcal{E}(s)/n) - \Sigma_{n,x}^{-1} (X^\top Z /n) \{ \hat{\mathcal{P}} - \mathcal{P} \} (Z^\top \mathcal{E}(s)/n) \\
    & \eqqcolon A_{12a} - A_{12b} - A_{12c}, \;\; \text{say.}
\end{align}
By Markov's inequality, it is easy to observe $||X^\top \mathcal{E}(s)/n|| \lesssim_p n^{-1/2}$.
Then, by \eqref{eq:sigma_X} and Assumption \ref{as:covmat}, we have $||A_{12a}|| \lesssim_p \sqrt{KL} / (\nu_{KL}^2 n)$.
Similarly, for $A_{12b}$, noting that $\rho_{\max}(\hat{\mathcal{P}}) \le \rho_{\max}((Z^\top Z / n)^{-}) \lesssim 1$,
\begin{align}
    || A_{12b} ||^2 
    & \lesssim_p  \frac{KL}{\nu_{KL}^4 n} \cdot \text{tr}\left\{ (X^\top Z /n) \hat{\mathcal{P}} (Z^\top \mathcal{E}(s)/n) (\mathcal{E}(s)^\top Z/n) \hat{\mathcal{P}} (Z^\top X /n) \right\}\\
    & \lesssim_p \frac{KL}{\nu_{KL}^4 n} \cdot \left\| (X^\top Z /n) (Z^\top \mathcal{E}(s)/n) \right\|^2.
\end{align}
Hence, it also holds that $||A_{12b}|| \lesssim_p \sqrt{KL} / (\nu_{KL}^2 n)$ by Assumptions \ref{as:error}(ii) and (iii).

Here, note that
\begin{align}
    \left\| \hat{\mathcal{P}} - \mathcal{P} \right\| \lesssim_p \frac{\sqrt{KL}} {\nu_{KL}^2 \sqrt{n}}
\end{align}
by Lemmas \ref{lem:matLLN1}(i) and \ref{lem:matLLN2}(ii).
Thus, for $A_{12c}$, we have
\begin{align}
    || A_{12c} ||^2
    & \lesssim \left\| \Sigma_{n,x}^{-1} (X^\top Z /n) \{ \hat{\mathcal{P}} - \mathcal{P} \} \right\|^2 \cdot \left\| Z^\top \mathcal{E}(s)/n \right\|^2
\end{align}
implying that $|| A_{12c} || \lesssim_p L\sqrt{K} / (\nu_{KL}^2 n)$.
Combining these results, we obtain the desired result under $L\sqrt{K} / (\nu_{KL}^2 \sqrt{n}) \lesssim 1$.

\bigskip

(ii) Note that $\bar R_x^\top M_z \bar R_x / n = \bar R^\top (M_z - M_x) \bar R / n$.
Then,
\begin{align}
    \left\| \bar R_x^\top M_z \bar R_x / n - \bbE \bar R_x^\top M_z \bbE \bar R_x \right\|
    & \le \left\| \bar R^\top M_z \bar R / n -  \bbE \bar R^\top M_z \bbE \bar R / n  \right\| \\
    & \quad + \left\| \bar R^\top M_x \bar R / n -  \bbE \bar R^\top M_x \bbE \bar R / n  \right\| = o_P(1)
\end{align}
by Lemma \ref{lem:matLLN2}(i).
With this, Assumption \ref{as:illposedness}(ii) implies that $\rho_{\min}(\bar R_x^\top M_z \bar R_x / n) \ge c\nu_{KL}$ with probability approaching one for some $1 > c > 0$.
Thus, by Weyl's inequality,
\begin{align}
    \rho_{\min}\left( \bar R_x^\top M_z \bar R_x / n + \lambda D  \right) \ge c[\nu_{KL} + \lambda \rho_D]
\end{align}
with probability approaching one.
Now, decompose
\begin{align}
    \hat \theta_n(s) - \theta_0(s)
    & = \left[ \bar R_x^\top M_z \bar R_x + \lambda D n \right]^{-1} \bar R_x^\top M_z Q(s) - \theta_0(s) \\
    & = \left[ \bar R_x^\top M_z \bar R_x + \lambda D n \right]^{-1} \bar R_x^\top M_z \bar R \theta_0(s)  - \theta_0(s) \\
    & \quad + \left[ \bar R_x^\top M_z \bar R_x + \lambda D n \right]^{-1} \bar R_x^\top M_z \mathcal{E}(s) + \left[ \bar R_x^\top M_z \bar R_x + \lambda D n \right]^{-1} \bar R_x^\top M_z U(s) \\
    & \eqqcolon B_1 + B_2 + B_3, \;\; \text{say}.
\end{align}
Noting that $R_x^\top M_z \bar R_x = R_x^\top M_z \bar R$ and
\begin{align}
    B_1 
    & = \left[ \bar R_x^\top M_z \bar R_x + \lambda D n \right]^{-1} \left[ \bar R_x^\top M_z \bar R + \lambda D n -  \lambda D n \right] \theta_0(s) - \theta_0(s) \\
    & = - \lambda \hat \Sigma_{n,r,\lambda}^{-1} D \theta_0(s),
\end{align}
where $\hat \Sigma_{n,r,\lambda} \coloneqq \bar R_x^\top M_z \bar R_x / n + \lambda D$,
\begin{align}
    ||B_1||^2 
    = \lambda^2 \cdot \theta_0(s)^\top D \hat \Sigma_{n,r,\lambda}^{-2} D \theta_0(s) 
    \lesssim_p \frac{\lambda^2 \left\|\theta_0(s)\right\|^2_D}{(\nu_{KL} + \lambda \rho_D)^2}.
\end{align}

Here, for any matrices $A$ and $B$ such that $A^\top A$ is nonsingular and $B$ is symmetric and positive semidefinite, it holds that\footnote{
    The symmetricity can be confirmed by $A\left([A^\top A + B]^{-1} B [A^\top A]^{-1}\right)A^\top = A\left([A^\top A]^{-1} B [A^\top A + B]^{-1}\right)A^\top$.
    }
\begin{align}
    A [A^\top A]^{-1} A^\top - A [A^\top A + B]^{-1} A^\top
    & = A\left([A^\top A]^{-1} - [A^\top A + B]^{-1}\right)A^\top \\
    & = A\left([A^\top A + B]^{-1} \left[ (A^\top A + B) - (A^\top A) \right][A^\top A]^{-1}\right)A^\top \\
    & = \underset{\text{symmetric positive semidefinite}}{A\left([A^\top A + B]^{-1} B [A^\top A]^{-1}\right)A^\top},
\end{align}
which implies that
\begin{align}\label{eq:fact1}
    \rho_{\max}(A [A^\top A + B]^{-1} A^\top) \le 1.
\end{align}

From this with $A = M_z \bar R_x/\sqrt{n}$ and $B = \lambda D$, we can easily observe that $||B_3|| \lesssim_p \ell_K(s)/\sqrt{\nu_{KL} + \lambda \rho_D}$.
For $B_2$, decompose $B_2 \coloneqq B_{21} + B_{22} + B_{23}$, where
\begin{align}
    B_{21}
    & \coloneqq \Sigma_{n, r, \lambda}^{-1} \bbE \bar R_x^\top M_z \mathcal{E}(s) / n\\
    B_{22}
    & \coloneqq \left\{ \hat \Sigma_{n,r,\lambda}^{-1} - \Sigma_{n, r, \lambda}^{-1} \right\} \bbE \bar R_x^\top M_z \mathcal{E}(s) / n \\
    B_{23}
    & \coloneqq \hat \Sigma_{n,r,\lambda}^{-1} \left\{ \bar R_x - \bbE \bar R_x \right\}^\top M_z \mathcal{E}(s) / n.
\end{align}
Note that \eqref{eq:fact1} implies the following:
\small\begin{align}\label{eq:fact2}
    \begin{split}
   \rho_{\max}\left( [A^\top A + B]^{-1}A^\top A[A^\top A + B]^{-1}\right)
   & = \rho_{\max}\left( [A^\top A + B]^{-1}A^\top A (A^\top A)^{-1} A^\top A [A^\top A + B]^{-1}\right) \\
   & = \rho_{\max}\left( (A^\top A)^{-1/2} A^\top A [A^\top A + B]^{-2} A^\top A  (A^\top A)^{-1/2} \right) \\
   & \le \rho_{\max}\left([A^\top A + B]^{-1}\right).
    \end{split}
\end{align}\normalsize
Using this, we obtain $||B_{21}|| \lesssim_p \sqrt{K}/\sqrt{n(\nu_{KL} + \lambda \rho_D)}$.
Next, by the same argument as in Lemma \ref{lem:matLLN2}(ii), we have
\begin{align}
    \left\| \hat \Sigma_{n,r,\lambda}^{-1} - \Sigma_{n, r, \lambda}^{-1}\right\| \lesssim_p \frac{\sqrt{KL}}{\sqrt{n}(\nu_{KL} + \lambda \rho_D )^2} \to 0.
\end{align}
Thus, by Markov's inequality, $||B_{22}|| = o_P(\sqrt{K}/\sqrt{n})$.
Finally, it is straightforward to observe that
\begin{align}
    ||B_{23}||
    & \le \left\|\hat \Sigma_{n,r,\lambda}^{-1} \left\{ \bar R_x^\top Z / n - \bbE \bar R_x^\top Z  /n \right\} \right\| \cdot \left\| (Z^\top Z / n)^{-1} Z^\top \mathcal{E}(s) / n \right\| \\
    & \lesssim_p \left( \frac{\sqrt{KL}}{\sqrt{n}(\nu_{KL} + \lambda \rho_D)} \right) \cdot \left( \frac{\sqrt{L}}{\sqrt{n}} \right) = o\left(\frac{\sqrt{K}}{\sqrt{n(\nu_{KL} + \lambda \rho_D)}}\right)
\end{align}
by Lemma \ref{lem:matLLN1}(i), where the last equality is due to $L/\sqrt{n (\nu_{KL} + \lambda \rho_D)} = o(1)$.
Combining these results, we obtain
\begin{align}\label{eq:theta_conv}
    \left\| \hat \theta_n(s) - \theta_0(s)\right\| \lesssim_p \frac{\sqrt{K}/\sqrt{n} + \ell_K(s)}{\sqrt{\nu_{KL} + \lambda \rho_D}} + \frac{\lambda ||\theta_0(s)||_D}{\nu_{KL} + \lambda \rho_D }.
\end{align}

By the triangle inequality,
\begin{align}
    \left\| \hat \alpha_n(\cdot, s) - \alpha_0(\cdot, s) \right\|_{L^2} 
    & \le \left\| \bm{\phi}_K(\cdot)^\top \theta_0(s) - \alpha_0(\cdot, s) \right\|_{L^2} + \left\| \bm{\phi}_K(\cdot)^\top (\hat \theta_n(s) - \theta_0(s)) \right\|_{L^2} \\
    & \lesssim \ell_K(s) + \left\| \bm{\phi}_K(\cdot)^\top (\hat \theta_n(s) - \theta_0(s)) \right\|_{L^2}.
\end{align}
Under Assumption \ref{as:base}(iii), we have
\begin{align}
    \left\| \bm{\phi}_K(\cdot)^\top (\hat \theta_n(s) - \theta_0(s)) \right\|_{L^2}^2
    & =   (\hat \theta_n(s) - \theta_0(s))^\top \left\{ \int_0^1 \bm{\phi}_K(t) \bm{\phi}_K(t)^\top \text{d}t \right\}  (\hat \theta_n(s) - \theta_0(s)) \\
    & \lesssim \left\| \hat \theta_n(s) - \theta_0(s)\right\|^2.
\end{align}
Then, the proof is completed in view of \eqref{eq:theta_conv}.
\qed


\begin{flushleft}
    \textbf{Proof of Theorem \ref{thm:normality}}
\end{flushleft}

(i) Using the notation in the proof of Theorem \ref{thm:conv}(i), we have $\sqrt{n}(\hat \beta_n(s) - \beta_0(s)) = \sqrt{n} A_1 + \sqrt{n} A_2$.
Recalling that $||A_2|| \lesssim_p \ell_K(s)$, $\sqrt{n}A_2 = o_P(1)$ holds by assumption.
Further, as shown in the proof of Theorem \ref{thm:conv}(i), $A_1 = A_{11} + A_{12a} - A_{12b} - A_{12c}$ with
\begin{align}
    ||A_{12a}|| \lesssim_p \frac{\sqrt{KL}}{\nu_{KL}^2 n}, \quad 
    ||A_{12b}|| \lesssim_p \frac{\sqrt{KL}}{\nu_{KL}^2 n}, \quad
    ||A_{12c}|| \lesssim_p \frac{L\sqrt{K}}{\nu_{KL}^2 n}.
\end{align}
Hence, $\sqrt{n} A_1 = \sqrt{n} A_{11} + o_P(1)$ under the assumptions made here.
Here, let $\mathbf{c}$ be an arbitrary $d_x \times 1$ vector such that $||\mathbf{c}|| = 1$, and let
\begin{align}
	a_n \coloneqq \mathbf{c}^\top \Lambda_{n,x}^{-1/2}\underbracket{\Sigma_{n,x}^{-1} \Psi_{n,x}^\top \mathcal{E}_n(s) / \sqrt{n}}_{\sqrt{n} A_{11}},
\end{align}
where $\Lambda_{n,x} \coloneqq \Sigma_{n,x}^{-1} \Omega_{n,x} \Sigma_{n,x}^{-1}$.
Below, we show that
\begin{align}
	a_n \overset{d}{\to} \mathcal{N}(0,1),
\end{align}
which implies the desired result.
Define 
\begin{align}
    \tilde a_{n,i} \coloneqq n^{-1/2}\mathbf{c}^\top \Lambda_{n,x}^{-1/2} \Sigma_{n,x}^{-1} (x_i - \tilde z_i)\varepsilon_i(s),
\end{align}
where $\tilde z_i \coloneqq (X^\top Z/n) \mathcal{P} z_i$.
Then, $a_n = \sum_{i=1}^n \tilde a_{n,i}$ holds with $\bbE \tilde a_{n,i} = 0$ and $\sum_{i=1}^n \bbE(\tilde a_{n,i}^2) = 1$.
Letting $\tilde a_{n,1,i} \coloneqq n^{-1/2}\mathbf{c}^\top \Lambda_{n,x}^{-1/2} \Sigma_{n,x}^{-1} x_i \varepsilon_i(s)$ and $\tilde a_{n,2,i} \coloneqq - n^{-1/2}\mathbf{c}^\top \Lambda_{n,x}^{-1/2} \Sigma_{n,x}^{-1} \tilde z_i\varepsilon_i(s)$ so that $\tilde a_{n,i} = \tilde a_{n,1,i} + \tilde a_{n,2,i}$, by the $c_r$-inequality, $\bbE (\tilde a_{n,i}^4) \le 8 \bbE ( \tilde a_{n,1,i}^4) + 8 \bbE ( \tilde a_{n,2,i}^4)$ holds, where
\begin{align}
    \bbE (\tilde a_{n,1,i}^4)
    \lesssim n^{-2} \mathbf{c}^\top \Lambda_{n,x}^{-1/2} \Sigma_{n,x}^{-1} x_i x_i^\top \Sigma_{n,x}^{-1} \Lambda_{n,x}^{-1/2} \mathbf{c}\mathbf{c}^\top \Lambda_{n,x}^{-1/2} \Sigma_{n,x}^{-1} x_i x_i^\top \Sigma_{n,x}^{-1} \Lambda_{n,x}^{-1/2} \mathbf{c} \lesssim n^{-2},
\end{align}
and
\begin{align}
    \bbE ( \tilde a_{n,2,i}^4)
    \lesssim n^{-2} \mathbf{c}^\top \Lambda_{n,x}^{-1/2} \Sigma_{n,x}^{-1} \tilde z_i \tilde z_i^\top \Sigma_{n,x}^{-1} \Lambda_{n,x}^{-1/2} \mathbf{c}\mathbf{c}^\top \Lambda_{n,x}^{-1/2} \Sigma_{n,x}^{-1} \tilde z_i \tilde z_i^\top \Sigma_{n,x}^{-1} \Lambda_{n,x}^{-1/2} \mathbf{c} \lesssim n^{-2} ||\tilde z_i||^4 \lesssim L^2/n^2
\end{align}
under Assumption \ref{as:error}(iii).
Hence, $\sum_{i=1}^n \bbE (\tilde a_{n,i}^4) \lesssim L^2/n \to 0$.
Then, applying Lyapunov's central limit theorem completes the proof.

\bigskip

(ii) First, by Weyl's inequality,
\small\begin{align}\label{eq:weyl_min}
    \rho_{\min}([A^\top A + B]^{-1}A^\top A[A^\top A + B]^{-1})
    & = \rho_{\min}([A^\top A + B]^{-1} - [A^\top A + B]^{-1}B[A^\top A + B]^{-1}) \\
    & \ge \rho_{\min}( [A^\top A + B]^{-1}) - \rho_{\max}( [A^\top A + B]^{-1}B[A^\top A + B]^{-1}).
\end{align}\normalsize
Then, setting $A = M_z \bbE \bar R_x/\sqrt{n}$ and $B = \lambda D$, by Assumption \ref{as:error}(iii) we have 
\small\begin{align}\label{eq:sigma}
    \begin{split}
    [\sigma_{n, \lambda}(t, s)]^2 
    & =  \bm{\phi}_K(t)^\top \Sigma_{n, r, \lambda}^{-1} \Omega_{n,r}(s) \Sigma_{n, r, \lambda}^{-1}\bm{\phi}_K(t) \\
    & \ge c_1 \cdot \bm{\phi}_K(t)^\top \left[ \bbE \bar R_x^\top M_z \bbE \bar R_x / n + \lambda D \right]^{-1} (\bbE \bar R_x^\top M_z \bbE \bar R_x /n) \left[ \bbE \bar R_x^\top M_z \bbE \bar R_x / n + \lambda D \right]^{-1} \bm{\phi}_K(t) \\
    & \ge \{c_2/(1 + \lambda) - c_3 \lambda / (\nu_{KL} + \lambda \rho_D)^2\} \cdot ||\bm{\phi}_K(t)||^2 \\
    & \ge c_4 \cdot ||\bm{\phi}_K(t)||^2
    \end{split}
\end{align}\normalsize
for a sufficiently large $n$ under the assumption $\lambda/\nu_{KL}^2 \to 0$, where the $c_1, \ldots, c_4$ are some fixed constants.

We can write
\begin{align}
    \frac{\sqrt{n}(\hat \alpha_n (t,s) - \alpha_0(t,s))}{\sigma_{n, \lambda}(t, s)}
    & = \frac{\sqrt{n}\bm{\phi}_K(t)^\top(B_1 + B_{21} + B_{22} + B_{23} + B_3)}{\sigma_{n, \lambda}(t, s)} + \frac{\sqrt{n}(\bm{\phi}_K(t)^\top \theta_0(s) - \alpha_0(t,s))}{\sigma_{n, \lambda}(t, s)}.
\end{align}
Recalling that
\begin{align}
    ||B_1|| \lesssim_p \frac{\lambda ||\theta_0(s)||_D}{\nu_{KL} + \lambda \rho_D}, \quad
    ||B_{22}|| \lesssim_p \frac{K\sqrt{L}}{n (\nu_{KL} + \lambda \rho_D)^2}, \quad 
    ||B_{23}|| \lesssim_p \frac{L\sqrt{K}}{n (\nu_{KL} + \lambda \rho_D)}, \quad
    ||B_3|| \lesssim_p \frac{\ell_K(s)}{\sqrt{\nu_{KL} + \lambda \rho_D}},
\end{align}
we can find that the dominant term of $\sqrt{n}(\hat \alpha_n (t,s) - \alpha_0(t,s))/\sigma_{n, \lambda}(t, s)$ is $\sqrt{n}\bm{\phi}_K(t)^\top B_{21}/\sigma_{n, \lambda}(t, s)$ considering \eqref{eq:sigma} under the assumptions introduced here.

Let
\begin{align}
    b_n 
    & \coloneqq [\sigma_{n, \lambda}(t, s)]^{-1} \bm{\phi}_K(t)^\top \underbracket{ \Sigma_{n, r, \lambda}^{-1} \bbE \bar R_x^\top M_z \mathcal{E}(s) / \sqrt{n}}_{\sqrt{n} B_{21}} \\
    \tilde b_{n,i}
    & \coloneqq n^{-1/2} [\sigma_{n, \lambda}(t, s)]^{-1} \bm{\phi}_K(t)^\top \Sigma_{n, r, \lambda}^{-1} (\bbE \bar R_x^\top Z/n) (Z^\top Z /n)^{-1} z_i \varepsilon_i(s)
\end{align}
such that $b_n = \sum_{i=1}^n \tilde b_{n,i}$, $\bbE \tilde b_{n,i} = 0$, and $\sum_{i=1}^n \bbE(\tilde b_{n,i}^2) = 1$ hold.
Observe that
\begin{align}
    \bbE (\tilde b_{n,i}^4)
    & \lesssim \frac{1}{n^2 [\sigma_{n, \lambda}(t, s)]^4} \left( z_i^\top (Z^\top Z /n)^{-1} (Z^\top \bbE \bar R_x /n) \Sigma_{n, r, \lambda}^{-1}\bm{\phi}_K(t) \bm{\phi}_K(t)^\top \Sigma_{n, r, \lambda}^{-1} (\bbE \bar R_x^\top Z/n) (Z^\top Z /n)^{-1} z_i \right)^2 \\
    & \lesssim \frac{1}{n^2} \left( z_i^\top (Z^\top Z /n)^{-1} (Z^\top \bbE \bar R_x /n) \Sigma_{n, r, \lambda}^{-2} (\bbE \bar R_x^\top Z/n) (Z^\top Z /n)^{-1} z_i \right)^2 \\
    & \lesssim \frac{||z_i||^2}{n^2}  z_i^\top (Z^\top Z /n)^{-1} (Z^\top \bbE \bar R_x /n) \Sigma_{n, r, \lambda}^{-2} (\bbE \bar R_x^\top M_z \bbE \bar R_x /n) \Sigma_{n, r, \lambda}^{-2} (\bbE \bar R_x^\top Z/n) (Z^\top Z /n)^{-1} z_i \\
    & \lesssim \frac{||z_i||^2}{n^2(\nu_{KL} + \lambda \rho_D)}  z_i^\top (Z^\top Z /n)^{-1} (Z^\top \bbE \bar R_x /n) \Sigma_{n, r, \lambda}^{-2}  (\bbE \bar R_x^\top Z/n) (Z^\top Z /n)^{-1} z_i \\
    & \lesssim \frac{||z_i||^2}{n^2(\nu_{KL} + \lambda \rho_D)^2}  z_i^\top (Z^\top Z /n)^{-1} (Z^\top \bbE \bar R_x /n) \Sigma_{n, r, \lambda}^{-1} (\bbE \bar R_x^\top Z/n) (Z^\top Z /n)^{-1} z_i \\
    & \lesssim \frac{||z_i||^2}{n^2(\nu_{KL} + \lambda \rho_D)^2}  z_i^\top (Z^\top Z /n)^{-1} z_i \\
    & \lesssim \frac{L^2}{n^2(\nu_{KL} + \lambda \rho_D)^2}
\end{align}
where we have used \eqref{eq:fact2} with $A = M_z \bbE \bar R_x/\sqrt{n}$ and $B = \lambda D$ in the fourth inequality and \eqref{eq:fact1} with $A = (Z^\top Z /n)^{-1/2} (Z^\top \bbE \bar R_x /n) $ and $B = \lambda D$ in the sixth inequality.
This suggests that $\sum_{i=1}^n \bbE (\tilde b_{n,i}^4) \to 0$, and the result follows from Lyapunov's central limit theorem.

\bigskip

(iii), (iv) We only prove that $Z^\top \hat{\bm{V}}_n(s) Z/n$ converges in probability to $Z^\top \bm{V}_n(s) Z/n$; the consistency of $X^\top \hat{\bm{V}}_n(s) X/n$ is analogous.
The consistency of the other parts are already proved in the preceding arguments.
Let $\tilde{\bm{V}}_n(s) \coloneqq \text{diag}\{\varepsilon_1^2(s), \ldots , \varepsilon_n^2(s)\}$.
By the triangle inequality,
\begin{align}
    \left\| Z^\top \hat{\bm{V}}_n(s) Z/n - Z^\top \bm{V}_n(s) Z/n\right\| \le \left\| Z^\top \left[ \hat{\bm{V}}_n(s) - \tilde{\bm{V}}_n(s) \right] Z / n \right\| + \left\| Z^\top \left[ \tilde{\bm{V}}_n(s) - \bm{V}_n(s) \right] Z / n \right\|.
\end{align}
Under Assumptions \ref{as:error}(ii) and (iii), by Markov's inequality, it is easy to observe that the second term on the right-hand side is of order $L/\sqrt{n}$ under Assumption \ref{as:error}(iii).

Write $\hat \varepsilon_i(s) = \varepsilon_i(s) + t_{n,r, i}(s) + t_{n,x,i}(s)$, where
\begin{align}
    t_{n,r,i}(s) \coloneqq \int_0^1 \bar q_i(t) [\alpha_0(t,s) - \check \alpha_n(t,s)] \text{d}t, \quad
    t_{n,x,i}(s) \coloneqq x_i^\top (\beta_0(s) - \hat \beta_n(s)),
\end{align}
where $\check \alpha_n(t,s) \coloneqq \bm{\phi}_K(t)^\top \check \theta_n(s)$.
By Theorem \ref{thm:conv}(i), we have $|t_{n,x,i}(s)| \lesssim_p n^{-1/2}$ uniformly in $i$.
For $t_{n,r,i}(s)$, noting that $|\int_0^1 \bar q_i(t) [\alpha_0(t,s) - \check \alpha_n(t,s)] \text{d}t| \lesssim ||\alpha_0(\cdot,s) - \check \alpha_n(\cdot,s)||_{L^2}$ by Cauchy--Schwarz inequality, Theorem \ref{thm:conv}(ii) gives $|t_{n,r,i}(s)| \lesssim_p \sqrt{K}/\sqrt{n \nu_{KL}}$ uniformly in $i$.
As $\hat \varepsilon_i^2(s) - \varepsilon_i^2(s) = t^2_{n,r, i}(s) + t^2_{n,x,i}(s) + 2t_{n,r,i}(s)\varepsilon_i(s) + 2t_{n,x,i}(s)\varepsilon_i(s) + 2t_{n,r,i}(s)t_{n,x,i}(s)$, we can decompose
\begin{align}
    Z^\top \left[ \hat{\bm{V}}_n(s) - \tilde{\bm{V}}_n(s) \right] Z / n = \gamma_{n1} + \gamma_{n2} + 2 \gamma_{n3} + 2 \gamma_{n4} + 2 \gamma_{n5},
\end{align}
where $\gamma_{n1} \coloneqq n^{-1}\sum_{i=1}^n z_i z_i^\top t^2_{n,r,i}(s)$, $\gamma_{n2} \coloneqq n^{-1}\sum_{i=1}^n z_i z_i^\top t^2_{n,x,i}(s)$, $\gamma_{n3} \coloneqq n^{-1}\sum_{i=1}^n z_i z_i^\top \varepsilon_i(s) t_{n,r,i}(s)$, $\gamma_{n4} \coloneqq n^{-1}\sum_{i=1}^n z_i z_i^\top \varepsilon_i(s) t_{n,x,i}(s)$, and $\gamma_{n5} \coloneqq n^{-1}\sum_{i=1}^n z_i z_i^\top t_{n,r,i}(s) t_{n,x,i}(s)$.
Then, by Markov's inequality, we have
\begin{align}
    &||\gamma_{n1}|| \lesssim_p K\sqrt{L}/(n \nu_{KL}), \quad
    ||\gamma_{n2}|| \lesssim_p \sqrt{L}/n, \quad
    ||\gamma_{n3}|| \lesssim_p \sqrt{KL}/\sqrt{n \nu_{KL}} \\
    &||\gamma_{n4}|| \lesssim_p \sqrt{L}/\sqrt{n}, \quad
    ||\gamma_{n5}|| \lesssim_p \sqrt{KL}/(n \sqrt{\nu_{KL}}).
\end{align}
This completes the proof.
\qed


\begin{lemma}\label{lem:quadraticCLT}
    Under the assumptions made in Theorem \ref{thm:test}, we have
    \begin{align}
        \frac{n B_{21}^\top \Phi_\mathcal{I} B_{21} - \mu_n}{\sqrt{v_n}} \overset{d}{\to} \mathcal{N}(0,1).
    \end{align}
\end{lemma}

\begin{proof}
    Recalling that $\Xi_n \coloneqq \Sigma_{n, r, \lambda}^{-1} \bbE (\bar R_x^\top Z) (Z^\top Z)^{-}$, observe 
\begin{align}
    \bbE[n B_{21}^\top \Phi_\mathcal{I} B_{21}] 
    & = \bbE[\mathcal{E}(s)^\top Z (Z^\top Z)^{-} \bbE( Z^\top \bar R_x) \Sigma_{n, r, \lambda}^{-1} \Phi_\mathcal{I} \Sigma_{n, r, \lambda}^{-1} \bbE (\bar R_x^\top Z) (Z^\top Z)^{-} Z^\top \mathcal{E}(s)]/n \\
    & = \text{tr}\left\{ \bbE[\Xi_n^\top \Phi_\mathcal{I} \Xi_n Z^\top \mathcal{E}(s) \mathcal{E}(s)^\top Z ]/n \right\} \\
    & = \underbracket{\text{tr}\left\{ \Xi_n^\top \Phi_\mathcal{I} \Xi_n (Z^\top \bm{V}_n(s) Z /n )\right\}}_{= \: \mu_n}.
\end{align}
Letting $\pi_{n,i,j} \coloneqq z_i^\top \Xi_n^\top \Phi_\mathcal{I} \Xi_n z_j$, we can write
\begin{align}
    n B_{21}^\top \Phi B_{21} 
    & = \frac{1}{n} \sum_{i = 1}^n \sum_{j = 1}^n \pi_{n,i,j} \varepsilon_i(s) \varepsilon_j(s) \\
    & = \frac{2}{n} \sum_{1 \le i < j \le n} \pi_{n,i,j} \varepsilon_i(s) \varepsilon_j(s) + \underbracket{\frac{1}{n} \sum_{i = 1}^n \pi_{n,i,i} \varepsilon^2_i(s)}_{= \: \text{tr}\left\{ \Xi_n^\top \Phi_\mathcal{I} \Xi_n (Z^\top \tilde{\bm{V}}_n(s) Z /n )\right\}}.
\end{align}
Here, we have
\begin{align}
    ||\Xi_n^\top \Phi_\mathcal{I} \Xi_n||^2
    & \lesssim \text{tr}\left\{ \Xi_n \Xi_n^\top \Xi_n \Xi_n^\top \right\} \\
    & \lesssim \text{tr}\left\{\Sigma_{n, r, \lambda}^{-1} (\bbE \bar R_x^\top M_z \bbE \bar R_x/n) \Sigma_{n, r, \lambda}^{-1} \Sigma_{n, r, \lambda}^{-1} (\bbE \bar R_x^\top M_z \bbE \bar R_x / n) \Sigma_{n, r, \lambda}^{-1} \right\} \lesssim \frac{K}{(\nu_{KL} + \lambda \rho_D)^2}
\end{align}
by \eqref{eq:fact2}.
Further, 
\begin{align}
    \left| \frac{1}{n} \sum_{i = 1}^n \pi_{n,i,i} \varepsilon^2_i(s) - \mu_n \right|
    & = \left| \text{tr}\left\{ \Xi_n^\top \Phi_\mathcal{I} \Xi_n \left(Z^\top \left[ \tilde{\bm{V}}_n(s) - \bm{V}_n(s) \right] Z /n \right)\right\} \right| \\
    & \le \left\|\Xi_n^\top \Phi_\mathcal{I} \Xi_n \right\| \cdot \left\| Z^\top \left[ \tilde{\bm{V}}_n(s) - \bm{V}_n(s) \right] Z /n \right\| \\
    & \lesssim_p \frac{L\sqrt{K}}{\sqrt{n}(\nu_{KL} + \lambda \rho_D)}
\end{align}
as in the proof of Theorem \ref{thm:normality}(iii), (iv).

Meanwhile, for a sufficiently large $n$,
\begin{align}\label{eq:v_order}
    \begin{split}
    v_n
    & = 2 \text{tr}\left\{ \Xi_n^\top \Phi_\mathcal{I} \Xi_n (Z^\top \bm{V}_n(s) Z /n ) \Xi_n^\top \Phi_\mathcal{I} \Xi_n (Z^\top \bm{V}_n(s) Z /n ) \right\} \\
    & \ge c_1  \text{tr}\left\{  (Z^\top \bm{V}_n(s) Z /n ) \Xi_n^\top \Xi_n (Z^\top \bm{V}_n(s) Z /n ) \Xi_n^\top \Xi_n \right\} \\
    & \ge c_2 \text{tr}\left\{ \Xi_n  (Z^\top Z /n) \Xi_n^\top \Xi_n (Z^\top Z /n) \Xi_n^\top \right\} \\
    & = c_2 \text{tr}\left\{ \Sigma_{n, r, \lambda}^{-1} (\bbE \bar R_x^\top M_z \bbE \bar R_x/n) \Sigma_{n, r, \lambda}^{-1} \Sigma_{n, r, \lambda}^{-1} (\bbE \bar R_x^\top M_z \bbE \bar R_x/n) \Sigma_{n, r, \lambda}^{-1} \right\} \\
    & \ge \{c_3/(1 + \lambda) - c_4 \lambda / (\nu_{KL} + \lambda \rho_D)^2\}^2 K \\
    & \ge c_5 K > 0
    \end{split}
\end{align}
for some constants $c_1, \ldots, c_5$, by Assumptions \ref{as:error}(iii) and \ref{as:misc}(ii) and \eqref{eq:weyl_min}.
These imply that
\begin{align}
    \frac{\frac{1}{n} \sum_{i = 1}^n \pi_{n,i,i} \varepsilon^2_i(s) - \mu_n}{\sqrt{v_n}} \lesssim_p \frac{L}{\sqrt{n}(\nu_{KL} + \lambda \rho_D)} \to 0.
\end{align}
Hence, we have
\begin{align}
    \frac{n B_{21}^\top \Phi_\mathcal{I} B_{21} - \mu_n}{\sqrt{v_n}} = \sum_{1 \le i < j \le n} \zeta_{n,i,j}  + o_P(1),
\end{align}
where $\zeta_{n,i,j} \coloneqq  (n \sqrt{v_n})^{-1} 2 \pi_{n,i,j}\varepsilon_i(s) \varepsilon_j(s)$.

To derive the limiting distribution of $\sum_{1 \le i < j \le n} \zeta_{n,i,j}$, we can use the central limit theorem for quadratic forms developed by \cite{de1987central}.
From Proposition 3.2 of \cite{de1987central}, if (1) $\text{Var}(\sum_{1 \le i < j \le n} \zeta_{n,i,j}) = 1 + o(1)$, (2) $G_{n,I} = o(1)$, (3) $G_{n,II} = o(1)$, and (4) $G_{n,IV} = o(1)$, we have $\sum_{1 \le i < j \le n} \zeta_{n,i,j} \overset{d}{\to} N(0,1)$, where
\begin{align}
	G_{n,I} & \coloneqq \sum_{1 \le i < j \le n} \bbE(\zeta_{n,i,j}^4) \\
	G_{n,II} & \coloneqq \sum_{1 \le i < j < k \le n} \bbE(\zeta_{n,i,j}^2\zeta_{n,ik}^2 + \zeta_{n,j,i}^2\zeta_{n,jk}^2 + \zeta_{n,k,i}^2\zeta_{n,k,j}^2) \\
	G_{n,IV} & \coloneqq \sum_{1 \le i < j < k < l\le n} \bbE(\zeta_{n,ij}\zeta_{n,i,k}\zeta_{n,lj}\zeta_{n,l,k} + \zeta_{n,i,j}\zeta_{n,i,l}\zeta_{n,k,j}\zeta_{n,k,l} + \zeta_{n,i,k}\zeta_{n,i,l}\zeta_{n,j,k}\zeta_{n,j,l}).
\end{align}
For (1), observe that
\begin{align}
    \text{Var}\left( \frac{2}{n} \sum_{1 \le i < j \le n} \pi_{n,i,j} \varepsilon_i(s) \varepsilon_j(s)\right)
    & = \frac{4}{n^2} \sum_{1 \le i < j \le n} \sum_{1 \le k < l \le n} \pi_{n,i,j} \pi_{n,k,l}  \bbE[ \varepsilon_i(s) \varepsilon_j(s) \varepsilon_k(s) \varepsilon_l(s)] \\
    & = \frac{4}{n^2} \sum_{1 \le i < j \le n} \pi^2_{n,i,j} \bbE[ \varepsilon^2_i(s)] \bbE[ \varepsilon^2_j(s)] \\
    & = \frac{2}{n^2} \sum_{i \neq j} \text{tr}\left\{  \Xi_n^\top \Phi_\mathcal{I} \Xi_n z_j z_j^\top \Xi_n^\top \Phi_\mathcal{I} \Xi_n z_i z_i^\top \right\} \bbE[ \varepsilon^2_i(s)] \bbE[ \varepsilon^2_j(s)] \\
    & = v_n - \frac{2}{n^2} \sum_{i = 1}^n \text{tr}\left\{ \Xi_n^\top \Phi_\mathcal{I} \Xi_n z_i z_i^\top \Xi_n^\top \Phi_\mathcal{I} \Xi_n z_i z_i^\top \right\} \left( \bbE[ \varepsilon^2_i(s)] \right)^2.
\end{align}
By easy calculation, we can find
\begin{align}
    \frac{2}{n^2} \sum_{i = 1}^n \text{tr}\left\{ \Xi_n^\top \Phi_\mathcal{I} \Xi_n z_i z_i^\top \Xi_n^\top \Phi_\mathcal{I} \Xi_n z_i z_i^\top \right\} \left( \bbE[ \varepsilon^2_i(s)] \right)^2
    & \lesssim \frac{KL}{n (\nu_{KL} + \lambda \rho_D)^2}.
\end{align}
Then, by \eqref{eq:v_order}, $\text{Var}(\sum_{1 \le i < j \le n} \zeta_{n,i,j}) \to 1$.

(2), (3), and (4) can be verified in the same manner as in the proof of Lemma A.11 of \cite{hoshino2022sieve}.
Indeed, the following results hold:
\begin{align}
    G_{n,I}
    & \lesssim  L^3/(n^2 K^2 [\nu_{KL} + \lambda \rho_D]^4) + L^4/(n^3 K^2 [\nu_{KL} + \lambda \rho_D]^4)\\
    G_{n,II}
    & \lesssim L^2/(n K^2 [\nu_{KL} + \lambda \rho_D]^4) + L^3/(n^2 K^2 [\nu_{KL} + \lambda \rho_D]^4) + L^4/(n^3 K^2 [\nu_{KL} + \lambda \rho_D]^4) \\
    G_{n,IV}
    & \lesssim 1/(K [\nu_{KL} + \lambda \rho_D]^2) + L^4/(n K^2 [\nu_{KL} + \lambda \rho_D]^4) + L^4/(n^2 K^2 [\nu_{KL} + \lambda \rho_D]^4) - G_{n,II}.
\end{align}
This completes the proof.
\end{proof}


\begin{flushleft}
    \textbf{Proof of Theorem \ref{thm:test}}
\end{flushleft}

Our test statistic is defined as a standardization of $T_n = n \int_\mathcal{I} \hat \alpha^2_n(t,s)\text{d}t$.
Trivially, under $\mathbb{H}_0$, we can write $T_n = n \int_\mathcal{I} (\hat \alpha_n(t,s) - \alpha_0(t,s))^2 \text{d}t$.
In view of \eqref{eq:v_order}, if we can verify that 
\begin{align}
T_n - n B_{21}^\top \Phi_\mathcal{I} B_{21} = o_P(\sqrt{K}),
\end{align}
the proof is completed by Lemma \ref{lem:quadraticCLT}.

Observe that
\begin{align}
    T_n 
    & = n \int_\mathcal{I} (\hat \alpha_n(t,s) - \alpha_0(t,s))^2 \text{d}t \\
    & = n \int_\mathcal{I} (\bm{\phi}_K(t)^\top [\hat \theta_n(s) - \theta_0(s)] + [\bm{\phi}_K(t)^\top \theta_0(s) - \alpha_0(t,s)])^2 \text{d}t \\
    & = n [\hat \theta_n(s) - \theta_0(s)]^\top \Phi_\mathcal{I} [\hat \theta_n(s) - \theta_0(s)] + 2 n \int_\mathcal{I} \bm{\phi}_K(t)^\top [\hat \theta_n(s) - \theta_0(s)] [\bm{\phi}_K(t)^\top \theta_0(s) - \alpha_0(t,s)] \text{d}t \\
    & \quad + n \underbracket{\int_\mathcal{I} (\bm{\phi}_K(t)^\top \theta_0(s) - \alpha_0(t,s))^2 \text{d}t}_{\le \; \ell_K^2(s)}. 
\end{align}
By Cauchy--Schwarz inequality,
\begin{align}
    & \left| \int_\mathcal{I} \bm{\phi}_K(t)^\top [\hat \theta_n(s) - \theta_0(s)] [\bm{\phi}_K(t)^\top \theta_0(s) - \alpha_0(t,s)] \text{d}t \right| \\
    & \le \ell_K(s) \left\| \bm{\phi}_K(t)^\top [\hat \theta_n(s) - \theta_0(s)] \right\|_{L^2} 
    \lesssim_p  \frac{\ell_K(s) \sqrt{K}}{\sqrt{n(\nu_{KL} + \lambda \rho_D)}} + \frac{\ell_K^2(s)}{\sqrt{\nu_{KL} + \lambda \rho_D}} + \frac{\lambda \ell_K(s) ||\theta_0(s)||_D}{\nu_{KL} + \lambda \rho_D }
\end{align}
as in the proof of Theorem \ref{thm:conv}(ii).

Next, using the decomposition in the proof of Theorem \ref{thm:conv}(ii), write
\begin{align}
    n [\hat \theta_n(s) - \theta_0(s)]^\top \Phi_\mathcal{I} [\hat \theta_n(s) - \theta_0(s)] = n \sum_{a = 1}^3 \sum_{b = 1}^3 B_a^\top \Phi_\mathcal{I} B_b.
\end{align}
We can easily observe the following results:
\begin{align}
    \left| B_1^\top \Phi_\mathcal{I} B_1 \right|
    \lesssim_p \frac{\lambda^2 ||\theta_0(s)||_D^2}{(\nu_{KL} + \lambda \rho_D)^2}, 
    \quad
    \left| B_3^\top \Phi_\mathcal{I} B_3 \right|
    \lesssim_p \frac{\ell_K^2(s)}{\nu_{KL} + \lambda \rho_D},
    \quad
    \left| B_1^\top \Phi_\mathcal{I} B_3 \right|
    \lesssim_p \frac{\lambda \ell_K(s)||\theta_0(s)||_D}{(\nu_{KL} + \lambda \rho_D)^{3/2}}.
\end{align}
Moreover, write $B_1^\top \Phi_\mathcal{I} B_2 = B_1^\top \Phi_\mathcal{I} (B_{21} + B_{22} + B_{23})$.
By similar calculations as above,
\small\begin{align}
    \left| B_1^\top \Phi_\mathcal{I} B_{21} \right|
    & = \lambda \left| \theta_0(s)^\top D \hat \Sigma_{n,r,\lambda}^{-1} \Phi_\mathcal{I} \Sigma_{n,r,\lambda}^{-1} \bbE \bar R_x^\top M_z \mathcal{E}(s)/n \right| \\
    & \le \lambda \left| \theta_0(s)^\top D \Sigma_{n,r,\lambda}^{-1} \Phi_\mathcal{I} \Sigma_{n,r,\lambda}^{-1} \bbE \bar R_x^\top M_z \mathcal{E}(s)/n \right|  + \lambda \left| \theta_0(s)^\top D \left\{ \hat \Sigma_{n,r,\lambda}^{-1} - \Sigma_{n,r,\lambda}^{-1} \right\}\Phi_\mathcal{I} \Sigma_{n,r,\lambda}^{-1} \bbE \bar R_x^\top M_z \mathcal{E}(s)/n \right| \\
    & \lesssim_p \frac{\lambda ||\theta_0(s)||_D}{\sqrt{n} (\nu_{KL} + \lambda \rho_D)^{3/2}} + \frac{\lambda ||\theta_0(s)||_D K \sqrt{L}}{n (\nu_{KL} + \lambda \rho_D)^{5/2}}.
\end{align}\normalsize
Similarly,
\begin{align}
    \left| B_1^\top \Phi_\mathcal{I} B_{22} \right|
    & = \lambda \left| \theta_0(s)^\top D \hat \Sigma_{n,r,\lambda}^{-1} \Phi_\mathcal{I} \left\{ \hat \Sigma_{n,r,\lambda}^{-1} - \Sigma_{n, r, \lambda}^{-1} \right\} \bbE \bar R_x^\top M_z \mathcal{E}(s) / n  \right| \lesssim_p \frac{\lambda ||\theta_0(s)||_D K \sqrt{L}}{n (\nu_{KL} + \lambda \rho_D)^3}.
\end{align}
and
\begin{align}
    \left| B_1^\top \Phi_\mathcal{I} B_{23} \right|
    & = \lambda \left| \theta_0(s)^\top D \hat \Sigma_{n,r,\lambda}^{-1} \Phi_\mathcal{I} \hat \Sigma_{n,r,\lambda}^{-1} \left\{ (\bar R_x^\top Z /n) - \bbE (\bar R_x^\top Z/n) \right\}^\top (Z^\top Z /n)^{-} Z^\top \mathcal{E}(s) / n \right| \\
    & \lesssim_p \frac{\lambda ||\theta_0(s)||_D L \sqrt{K} }{ n (\nu_{KL} + \lambda \rho_D)^2}.
\end{align}
We can also observe that
\begin{align}
    & \left| B_3^\top \Phi_\mathcal{I} B_{21} \right|
    \lesssim_p \frac{\ell_K(s)}{\sqrt{n}(\nu_{KL} + \lambda \rho_D)} + \frac{\ell_K(s)K \sqrt{L}}{n (\nu_{KL} + \lambda \rho_D)^{5/2}} + \frac{\ell_K(s) K \sqrt{L}}{n (\nu_{KL} + \lambda \rho_D)^{3/2}} \\
    &\left| B_3^\top \Phi_\mathcal{I} B_{22} \right|
    \lesssim_p \frac{\ell_K(s) K \sqrt{L}}{n (\nu_{KL} + \lambda \rho_D)^{5/2}}, 
    \quad
    \left| B_3^\top \Phi_\mathcal{I} B_{23} \right|
    \lesssim_p \frac{\ell_K(s) L \sqrt{K}}{n (\nu_{KL} + \lambda \rho_D)^{3/2}} \\
    & \left| B_{21}^\top \Phi_\mathcal{I} B_{22} \right| 
    \lesssim_p \frac{K\sqrt{KL}}{n^{3/2}(\nu_{KL} + \lambda \rho_D)^{5/2}},
    \quad 
    \left| B_{21}^\top \Phi_\mathcal{I} B_{23} \right|
    \lesssim_p \frac{KL}{n^{3/2}(\nu_{KL} + \lambda \rho_D)^{3/2}} \\
    & \left| B_{22}^\top \Phi_\mathcal{I} B_{22} \right| 
    \lesssim_p \frac{K^2 L}{n^2(\nu_{KL} + \lambda \rho_D)^4},
    \quad 
    \left| B_{23}^\top \Phi_\mathcal{I} B_{23} \right|
    \lesssim_p \frac{KL^2}{n^2(\nu_{KL} + \lambda \rho_D)^2},
    \quad 
    \left| B_{22}^\top \Phi_\mathcal{I} B_{23} \right|
    \lesssim_p \frac{(KL)^{3/2}}{n^2(\nu_{KL} + \lambda \rho_D)^3}.
\end{align}
Under the assumptions made, we can find that $T_n = n B_{21}^\top \Phi_\mathcal{I} B_{21} + o_p(\sqrt{K})$, as desired. 
\qed


\begin{flushleft}
    \textbf{Proof of Theorem \ref{thm:tilde}}
\end{flushleft}

Observe that
\begin{align}
	\hat Q(s) 
    & = W_n \int_0^1 Q(t) \alpha_0(t,s) \text{d}t + X \beta_0(s) + \mathcal{E}(s) + E(s)\\
    & = W_n \int_0^1 \hat Q(t) \alpha_0(t,s) \text{d}t - W_n \int_0^1 E(t) \alpha_0(t,s) \text{d}t + X \beta_0(s) + \mathcal{E}(s) + E(s)\\
    & = \hat R \theta_0(s) + X \beta_0(s) + \mathcal{E}(s) + E(s) - V(s) + \hat U(s),
\end{align}
where $E = (e_1, \ldots, e_n)^\top$, $e_i \coloneqq \hat q_i - q_i$, $V(s) = (v_1(s), \ldots, v_n(s))^\top$, $v_i(s) \coloneqq \sum_{j=1}^n w_{i,j} \int_0^1 e_j(t) \alpha_0(t,s) \text{d}t$, $\hat U(s) = (\hat u_1(s), \ldots, \hat u_n(s))^\top$, and $\hat u_i(s) \coloneqq \sum_{j = 1}^n w_{i,j} \int_0^1 \hat q_j(t) \alpha_0(t,s)\text{d}t - \hat r_i^\top \theta_0 $.
As shown in \eqref{eq:interp_err}, $|e_i| \lesssim \kappa^\xi$ uniformly.
From this, $|v_i(s)| \lesssim \kappa^\xi$ is straightforward.
Further, similar to \eqref{eq:bias}, we have $| \hat u_i(s) | \lesssim \ell_K(s)$ uniformly.
Write
\begin{align}
	\tilde \beta_n(s) - \beta_0(s) 
    & = \left[ X^\top (I_n - \hat S) X \right]^{-1} X^\top (I_n - \hat S) \hat Q(s) - \beta_0(s) \\
	& = \left[ X^\top (I_n - \hat S) X \right]^{-1} X^\top (I_n - \hat S) [\mathcal{E}(s) + \hat U(s) + E(s) - V(s)] \\
	& \eqqcolon A'_1 + A'_2 + A'_3 + A'_4, \;\; \text{say.} 
\end{align}
Applying Fact A.2 in \cite{hoshino2022sieve} and Lemmas \ref{lem:matLLN1}(ii) and \ref{lem:matLLN2}(iv), we have
\begin{align}
    & \left\| X^\top \hat S X /n - \bbE(X^\top \bar R/n) \left[\bbE \bar R^\top M_z \bbE \bar R /n \right]^{-1} \bbE(\bar R^\top X/n)  \right\| \\
    & \lesssim \rho_{\max}\left( (\hat R^\top X/n)  (X^\top \hat R/n) \right) \left\| \left[\hat R^\top M_z \hat R /n \right]^{-1} - \left[\bbE \bar R^\top M_z \bbE \bar R /n \right]^{-1} \right\| \\
    & \quad + \rho_{\max}\left(  \left[\bbE \bar R^\top M_z \bbE \bar R /n \right]^{-1} \right) \left\| \hat R^\top X/n - \bbE(\bar R^\top X/n) \right\| \lesssim_p \frac{\sqrt{KL}} {\nu_{KL}^2 \sqrt{n}} + \frac{\kappa^\xi \sqrt{KL}} {\nu_{KL}^2} + \frac{\sqrt{K}}{\nu_{KL} \sqrt{n}} + \frac{\kappa^\xi \sqrt{K}}{\nu_{KL}}.
\end{align}
This implies that $||X^\top (I_n - \hat S) X /n - \Sigma_{n,x}|| \lesssim_p \sqrt{KL}/(\nu_{KL}^2 \sqrt{n})$ and that $\rho\left( X^\top (I_n - \hat S) X /n \right) > 0$ with probability approaching one.
Then, by the same argument as in the proof of Theorem \ref{thm:conv}(i), we can easily observe that $||A_2'|| \lesssim_p \ell_K(s)$, $||A_3'|| \lesssim_p \kappa^\xi$, and $||A_4'|| \lesssim_p \kappa^\xi$ hold.

For $A_1'$, decompose $A_1' = A_1 + A_{12}' + A_{13}'$, where 
\begin{align}
    A_1
    & = \left[ X^\top (I_n - S) X /n \right]^{-1} X^\top (I_n -  S) \mathcal{E}(s) /n \\
    A_{12}'
    & = \left[ X^\top (I_n - S) X /n \right]^{-1} X^\top (S - \hat S) \mathcal{E}(s) /n \\
    A_{13}'
    & = \left\{ \left[ X^\top (I_n - \hat S) X /n \right]^{-1} - \left[ X^\top (I_n - S) X /n \right]^{-1} \right\} X^\top (I_n - \hat S) \mathcal{E}(s) /n.  
\end{align}
Write 
\small\begin{align}
    & X^\top (S - \hat S) \mathcal{E}(s)/n \\
    & = (X^\top Z/n) (Z^\top Z/n)^{-} \bigl\{ \underbracket{(Z^\top \bar R /n)[\bar R^\top M_z \bar R/n]^{-1} (\bar R^\top Z/n) - (Z^\top \hat R/n) [\hat R^\top M_z \hat R/n]^{-1} (\hat R^\top Z/n)}_{\eqqcolon \: \bm{L}} \bigr\} (Z^\top Z/n)^{-} Z^\top \mathcal{E}(s)/n.
\end{align}\normalsize
By the same argument as above, we observe that
\begin{align}
    \left\| \bm{L} \right\| 
    & \lesssim \rho_{\max}\left( (\bar R^\top Z/n) (Z^\top \bar R/n) \right) \left\| [\bar R^\top M_z \bar R/n]^{-1} - [\hat R^\top M_z \hat R/n]^{-1} \right\| \\
    & \quad + \rho_{\max}\left(  \left[\hat R^\top M_z \hat R /n \right]^{-1} \right) \left\| \bar R^\top Z/n - \hat R^\top Z/n \right\| \lesssim_p \frac{\kappa^\xi \sqrt{KL}} {\nu_{KL}^2} + \frac{\kappa^\xi \sqrt{KL}}{\nu_{KL}}.
\end{align}
Thus,
\begin{align}
    ||A_{12}'||^2
    & = \mathcal{E}(s)^\top (S - \hat S) X \left[ X^\top (I_n - S) X /n \right]^{-2} X^\top (S - \hat S) \mathcal{E}(s) /n^2 \\
    & \lesssim_p (\mathcal{E}(s)^\top Z /n) (Z^\top Z/n)^{-} \bm{L} (Z^\top Z/n)^{-} (Z^\top X/n) (X^\top Z/n) (Z^\top Z/n)^{-} \bm{L} (Z^\top Z/n)^{-} (Z^\top \mathcal{E}(s) /n) \\
    & \lesssim_p ||\bm{L}||^2 \cdot ||Z^\top \mathcal{E}(s) /n||^2,
\end{align}
which yields $||A_{12}'|| \lesssim_p \kappa^\xi L \sqrt{K}/(\sqrt{n} \nu_{KL}^2)$.
Similarly,
\begin{align}
    ||A_{13}'||
    & \lesssim_p \left\|\left[ X^\top (I_n - \hat S) X /n \right]^{-1} - \left[ X^\top (I_n - S) X /n \right]^{-1} \right\| \\
    & \quad \times \left\{ ||X^\top \mathcal{E}(s) /n|| + || X^\top Z (Z^\top Z)^{-} Z \hat R [\hat R^\top M_z \hat R ]^{-1} \hat R^\top Z (Z^\top Z)^{-} Z^\top \mathcal{E}(s) /n||\right\} \\
    & \lesssim_p \frac{\kappa^\xi \sqrt{KL}}{\sqrt{n} \nu_{KL}^2} + \frac{\kappa^\xi L \sqrt{K}}{\sqrt{n} \nu_{KL}^2}.
\end{align}
Combining all these results, we have
\begin{align}
    \tilde \beta_n(s) - \beta_0(s) - A_1
    & \lesssim_p \frac{\kappa^\xi L \sqrt{K}}{\sqrt{n} \nu_{KL}^2} + \ell_K(s) + \kappa^\xi = o_P(n^{-1/2}),
\end{align}
which implies that $\tilde \beta_n(s)$ and $\hat \beta_n(s)$ have the same asymptotic distribution.

\bigskip

Next, similar to the above discussion, decompose
\begin{align}
    \tilde \theta_n(s) - \theta_0(s)
    & = \left[ \hat R_x^\top M_z \hat R_x + \lambda D n \right]^{-1} \hat R_x^\top M_z \hat Q(s) - \theta_0(s) \\
    & = \left[ \hat R_x^\top M_z \hat R_x + \lambda D n \right]^{-1} \hat R_x^\top M_z \hat R \theta_0(s) - \theta_0(s) \\
    & \quad + \left[ \hat R_x^\top M_z \hat R_x + \lambda D n \right]^{-1} \hat R_x^\top M_z [\mathcal{E}(s) + \hat U(s) + E(s) - V(s)] \\
    & \eqqcolon B_1' + B_2' + B_3' + B_4' + B_5', \;\; \text{say}.
\end{align}
By Lemma \ref{lem:matLLN2}(iii), we have $\rho_{\max}([\hat R_x^\top M_z \hat R_x/n + \lambda D]^{-1}) \lesssim_p (\nu_{KL} + \lambda \rho_D)^{-1}$.
Then, for $B_1'$, noting that $B_1' = - \lambda [\hat R_x^\top M_z \hat R_x /n + \lambda D]^{-1} D \theta_0(s)$, we have
\begin{align}
    ||B_1'|| \lesssim_p \frac{\lambda || \theta_0(s)||_D}{\nu_{KL} + \lambda \rho_D}.
\end{align}
Additionally, we can easily find that $||B_3'|| \lesssim_p \ell_K(s) / \sqrt{\nu_{KL} + \lambda \rho_D}$, $||B_4'|| \lesssim_p \kappa^\xi / \sqrt{\nu_{KL} + \lambda \rho_D}$, and $||B_5'|| \lesssim_p \kappa^\xi/\sqrt{\nu_{KL} + \lambda \rho_D}$ hold.

For $B_2'$, decompose it further as $B_2' = B_2 + B_{22}' + B_{23}'$, where 
\begin{align}
    B_2
    & = \left[ \bar R_x^\top M_z \bar R_x /n + \lambda D  \right]^{-1} \bar R_x^\top M_z \mathcal{E}(s) /n \\
    B_{22}'
    & = \left[ \bar R_x^\top M_z \bar R_x /n + \lambda D  \right]^{-1} (\hat R_x - \bar R_x)^\top M_z \mathcal{E}(s) /n \\
    B_{23}'
    & = \left\{ \left[ \hat R_x^\top M_z \hat R_x /n + \lambda D  \right]^{-1} - \left[ \bar R_x^\top M_z \bar R_x /n + \lambda D  \right]^{-1} \right\} \hat R_x^\top M_z \mathcal{E}(s) /n.  
\end{align}
By Lemma \ref{lem:matLLN1}(ii), we have
\begin{align}
    ||B_{22}'||
    & = \left\| \left[ \bar R_x^\top M_z \bar R_x /n + \lambda D  \right]^{-1} (\hat R Z /n - \bar R Z/n)^\top  (Z^\top Z /n)^{-} Z^\top \mathcal{E}(s) /n\right\|\\
    & \quad + \left\| \left[ \bar R_x^\top M_z \bar R_x /n + \lambda D  \right]^{-1} (\hat R X /n - \bar R X /n)^\top  (X^\top X /n)^{-1} X^\top \mathcal{E}(s) /n\right\| \\
    & \lesssim_p \frac{\kappa^\xi L \sqrt{K}}{\sqrt{n}(\nu_{KL} + \lambda \rho_D)} + \frac{\kappa^\xi \sqrt{K}}{\sqrt{n}(\nu_{KL} + \lambda \rho_D)}.
\end{align}
It also holds that
\begin{align}
    ||B_{23}'||
    & \le \left\| \left[ \hat R_x^\top M_z \hat R_x /n + \lambda D  \right]^{-1} - \left[ \bar R_x^\top M_z \bar R_x /n + \lambda D  \right]^{-1} \right\| \cdot \left\| \hat R_x^\top M_z \mathcal{E}(s) /n \right\| \\
    & \lesssim_p \frac{\kappa^\xi L \sqrt{K}}{\sqrt{n} (\nu_{KL} + \lambda \rho_D)^2}.
\end{align}
Hence, we have
\begin{align}
    & \frac{\sqrt{n}(\tilde \alpha_n(t,s) - \alpha_0(t,s) - \bm{\phi}_K(t)^\top B_2)}{\sigma_{n,\lambda}(t,s)} \\
    & = \frac{\sqrt{n} \bm{\phi}_K(t)^\top (B_1' + B_{22}' + B_{23}' + B_3' + B_4' + B_5')}{\sigma_{n,\lambda}(t,s)} + \frac{\sqrt{n} (\bm{\phi}_K(t)^\top \theta_0(s) - \alpha_0(t,s))}{\sigma_{n,\lambda}(t,s)} \\
    & \lesssim_p \frac{\sqrt{n} \lambda ||\theta_0(s)||_D}{\nu_{KL} + \lambda \rho_D} + \frac{\kappa^\xi L \sqrt{K}}{(\nu_{KL} + \lambda \rho_D)^2} + \frac{\sqrt{n}(\ell_K(s) + \kappa^\xi)}{\sqrt{\nu_{KL} + \lambda \rho_D}} + \frac{\sqrt{n}|\bm{\phi}_K(t)^\top \theta_0(s) - \alpha_0(t,s)|}{||\bm{\phi}_K(t)||} = o(1),
\end{align}
implying the desired result.

\section{Supplementary simulation results}\label{sec:app_MC}

This section provides the detailed simulation results under incompletely observed outcome functions.
As described in the main text, we suppose that for each unit $i$ we can observe $\{(s_{i,j}, q_i(y_{i,j}))\}_{j = 1}^m$, where $s_{i,j}$'s are uniformly randomly drawn from $[0,1]$, for $m \in \{15, 50\}$.
To recover the entire functional form of the $q_i$ function, we apply the linear interpolation method in \eqref{eq:interp}.
The simulation scenarios examined are all identical to those used in the main text.
The results for the 2SLS estimator and the Wald test are summarized in Tables \ref{tab:estimation_interp} and \ref{tab:test_interp}, respectively.

\begin{table}[ht]
    \caption{Estimation performance}
    \label{tab:estimation_interp}
    \centering\footnotesize
    \begin{tabular}{llll|cccccccccc}
        \hline \hline
        &   &   &   & \multicolumn{2}{c}{$\beta$} & \multicolumn{2}{c}{$\alpha$ ($\lambda_c = 0.5$)} & \multicolumn{2}{c}{$\alpha$ ($\lambda_c = 1$)} & \multicolumn{2}{c}{$\alpha$ ($\lambda_c = 2$)} & \multicolumn{2}{c}{$\alpha$ ($\lambda_c = 3$)} \\
    DGP & $m$ & $n$ & \# knots & BIAS   & RMSE  & BIAS               & RMSE  & BIAS             & RMSE  & BIAS             & RMSE  & BIAS             & RMSE  \\
    \hline
    1 & 15 & 400 & 2 & 0.0002  & 0.0376  & 0.0214  & 0.1071  & 0.0228  & 0.1036  & 0.0203  & 0.1022  & 0.0162  & 0.1010  \\
    &  &  & 3 & 0.0004  & 0.0385  & 0.0221  & 0.1199  & 0.0230  & 0.1157  & 0.0197  & 0.1136  & 0.0149  & 0.1119  \\
    &  & 1600 & 2 & 0.0008  & 0.0184  & 0.0161  & 0.0974  & 0.0201  & 0.0962  & 0.0223  & 0.0978  & 0.0220  & 0.0991  \\
    &  &  & 3 & 0.0007  & 0.0184  & 0.0170  & 0.1118  & 0.0208  & 0.1097  & 0.0226  & 0.1105  & 0.0220  & 0.1112  \\
    & 50 & 400 & 2 & -0.0010  & 0.0370  & 0.0202  & 0.1060  & 0.0214  & 0.1025  & 0.0189  & 0.1010  & 0.0148  & 0.0997  \\
    &  &  & 3 & -0.0008  & 0.0372  & 0.0208  & 0.1187  & 0.0217  & 0.1147  & 0.0184  & 0.1125  & 0.0136  & 0.1108  \\
    &  & 1600 & 2 & -0.0004  & 0.0176  & 0.0152  & 0.0959  & 0.0190  & 0.0951  & 0.0211  & 0.0968  & 0.0208  & 0.0980  \\
    &  &  & 3 & -0.0004  & 0.0177  & 0.0161  & 0.1102  & 0.0197  & 0.1086  & 0.0214  & 0.1094  & 0.0208  & 0.1102  \\
   2 & 15 & 400 & 2 & -0.0001  & 0.0377  & -0.0023  & 0.0917  & -0.0065  & 0.0829  & -0.0134  & 0.0769  & -0.0196  & 0.0742  \\
    &  &  & 3 & 0.0002  & 0.0388  & -0.0023  & 0.1074  & -0.0067  & 0.0990  & -0.0143  & 0.0931  & -0.0211  & 0.0903  \\
    &  & 1600 & 2 & 0.0007  & 0.0184  & -0.0001  & 0.0848  & -0.0024  & 0.0808  & -0.0060  & 0.0774  & -0.0091  & 0.0757  \\
    &  &  & 3 & 0.0007  & 0.0184  & 0.0001  & 0.1015  & -0.0023  & 0.0976  & -0.0061  & 0.0942  & -0.0094  & 0.0925  \\
    & 50 & 400 & 2 & -0.0012  & 0.0375  & -0.0011  & 0.0921  & -0.0055  & 0.0836  & -0.0125  & 0.0775  & -0.0188  & 0.0748  \\
    &  &  & 3 & -0.0012  & 0.0376  & -0.0011  & 0.1078  & -0.0058  & 0.0996  & -0.0134  & 0.0937  & -0.0203  & 0.0908  \\
    &  & 1600 & 2 & -0.0005  & 0.0179  & 0.0013  & 0.0851  & -0.0012  & 0.0813  & -0.0049  & 0.0780  & -0.0081  & 0.0763  \\
    &  &  & 3 & -0.0005  & 0.0181  & 0.0014  & 0.1018  & -0.0011  & 0.0981  & -0.0051  & 0.0948  & -0.0085  & 0.0930  \\
   3 & 15 & 400 & 2 & -0.0004  & 0.0401  & 0.0074  & 0.1841  & 0.0100  & 0.1964  & 0.0076  & 0.2076  & 0.0027  & 0.2125  \\
    &  &  & 3 & -0.0004  & 0.0410  & 0.0079  & 0.1870  & 0.0099  & 0.2008  & 0.0065  & 0.2125  & 0.0008  & 0.2172  \\
    &  & 1600 & 2 & 0.0007  & 0.0195  & 0.0006  & 0.1559  & 0.0061  & 0.1708  & 0.0096  & 0.1885  & 0.0098  & 0.1977  \\
    &  &  & 3 & 0.0007  & 0.0197  & 0.0016  & 0.1552  & 0.0067  & 0.1732  & 0.0098  & 0.1927  & 0.0094  & 0.2024  \\
    & 50 & 400 & 2 & -0.0014  & 0.0395  & 0.0142  & 0.1823  & 0.0171  & 0.1946  & 0.0148  & 0.2061  & 0.0100  & 0.2112  \\
    &  &  & 3 & -0.0014  & 0.0399  & 0.0148  & 0.1852  & 0.0170  & 0.1993  & 0.0138  & 0.2116  & 0.0081  & 0.2165  \\
    &  & 1600 & 2 & -0.0005  & 0.0192  & 0.0071  & 0.1545  & 0.0127  & 0.1682  & 0.0165  & 0.1862  & 0.0168  & 0.1957  \\
    &  &  & 3 & -0.0005  & 0.0192  & 0.0081  & 0.1533  & 0.0134  & 0.1707  & 0.0167  & 0.1907  & 0.0165  & 0.2009  \\ 
    \hline \hline
    \end{tabular}
\end{table}

\begin{table}[ht]
    \caption{Rejection frequency}
    \label{tab:test_interp}
    \centering\footnotesize
    \begin{tabular}{llll|ccccccccc}
        \hline \hline
      &      &    &    & \multicolumn{3}{c}{$\varrho = 0$} & \multicolumn{3}{c}{$\varrho = 0.1$} & \multicolumn{3}{c}{$\varrho = 0.2$}\\
    $n$ & \# knots & $m$ & $\lambda_c$ & 10\%  & 5\%   & 1\%   & 10\%  & 5\%     & 1\%   & 10\%  & 5\%  & 1\%   \\
    \hline
    400 & 2 & 15 & 0.5 & 0.071  & 0.035  & 0.016  & 0.261  & 0.194  & 0.105  & 0.967  & 0.939  & 0.829  \\
    &  &  & 1 & 0.063  & 0.035  & 0.020  & 0.710  & 0.608  & 0.440  & 1.000  & 0.998  & 0.992  \\
    &  &  & 2 & 0.056  & 0.032  & 0.015  & 0.946  & 0.921  & 0.861  & 1.000  & 1.000  & 1.000  \\
    &  &  & 3 & 0.051  & 0.035  & 0.019  & 0.974  & 0.959  & 0.921  & 1.000  & 1.000  & 1.000  \\
    &  & 50 & 0.5 & 0.079  & 0.049  & 0.018  & 0.298  & 0.210  & 0.110  & 0.983  & 0.967  & 0.880  \\
    &  &  & 1 & 0.077  & 0.048  & 0.019  & 0.737  & 0.641  & 0.469  & 1.000  & 0.998  & 0.995  \\
    &  &  & 2 & 0.069  & 0.040  & 0.019  & 0.969  & 0.945  & 0.903  & 1.000  & 1.000  & 1.000  \\
    &  &  & 3 & 0.062  & 0.044  & 0.023  & 0.981  & 0.975  & 0.955  & 1.000  & 1.000  & 1.000  \\
    & 3 & 15 & 0.5 & 0.068  & 0.034  & 0.016  & 0.242  & 0.178  & 0.096  & 0.955  & 0.919  & 0.797  \\
    &  &  & 1 & 0.060  & 0.034  & 0.016  & 0.690  & 0.590  & 0.430  & 1.000  & 0.998  & 0.991  \\
    &  &  & 2 & 0.052  & 0.031  & 0.013  & 0.941  & 0.912  & 0.844  & 1.000  & 1.000  & 1.000  \\
    &  &  & 3 & 0.045  & 0.033  & 0.016  & 0.964  & 0.950  & 0.917  & 1.000  & 1.000  & 1.000  \\
    &  & 50 & 0.5 & 0.073  & 0.046  & 0.018  & 0.292  & 0.205  & 0.105  & 0.980  & 0.959  & 0.872  \\
    &  &  & 1 & 0.072  & 0.044  & 0.018  & 0.741  & 0.650  & 0.479  & 1.000  & 0.999  & 0.995  \\
    &  &  & 2 & 0.061  & 0.036  & 0.017  & 0.965  & 0.941  & 0.901  & 1.000  & 1.000  & 1.000  \\
    &  &  & 3 & 0.061  & 0.041  & 0.023  & 0.979  & 0.970  & 0.955  & 1.000  & 1.000  & 1.000  \\
   1600 & 2 & 15 & 0.5 & 0.058  & 0.043  & 0.023  & 0.505  & 0.364  & 0.188  & 0.998  & 0.996  & 0.991  \\
    &  &  & 1 & 0.057  & 0.043  & 0.023  & 0.936  & 0.875  & 0.692  & 1.000  & 1.000  & 0.999  \\
    &  &  & 2 & 0.063  & 0.040  & 0.022  & 0.997  & 0.997  & 0.990  & 1.000  & 1.000  & 1.000  \\
    &  &  & 3 & 0.064  & 0.041  & 0.022  & 0.999  & 0.999  & 0.998  & 1.000  & 1.000  & 1.000  \\
    &  & 50 & 0.5 & 0.078  & 0.054  & 0.028  & 0.611  & 0.449  & 0.243  & 1.000  & 1.000  & 1.000  \\
    &  &  & 1 & 0.078  & 0.053  & 0.029  & 0.970  & 0.942  & 0.791  & 1.000  & 1.000  & 1.000  \\
    &  &  & 2 & 0.077  & 0.052  & 0.028  & 1.000  & 1.000  & 0.998  & 1.000  & 1.000  & 1.000  \\
    &  &  & 3 & 0.078  & 0.052  & 0.026  & 1.000  & 1.000  & 1.000  & 1.000  & 1.000  & 1.000  \\
    & 3 & 15 & 0.5 & 0.056  & 0.041  & 0.023  & 0.501  & 0.353  & 0.176  & 0.999  & 0.996  & 0.990  \\
    &  &  & 1 & 0.055  & 0.041  & 0.022  & 0.940  & 0.869  & 0.695  & 1.000  & 1.000  & 0.999  \\
    &  &  & 2 & 0.062  & 0.039  & 0.021  & 0.998  & 0.997  & 0.990  & 1.000  & 1.000  & 1.000  \\
    &  &  & 3 & 0.061  & 0.040  & 0.020  & 1.000  & 0.999  & 0.998  & 1.000  & 1.000  & 1.000  \\
    &  & 50 & 0.5 & 0.074  & 0.051  & 0.027  & 0.605  & 0.439  & 0.236  & 1.000  & 1.000  & 1.000  \\
    &  &  & 1 & 0.074  & 0.050  & 0.028  & 0.972  & 0.942  & 0.808  & 1.000  & 1.000  & 1.000  \\
    &  &  & 2 & 0.076  & 0.050  & 0.028  & 1.000  & 1.000  & 0.999  & 1.000  & 1.000  & 1.000  \\
    &  &  & 3 & 0.077  & 0.049  & 0.026  & 1.000  & 1.000  & 1.000  & 1.000  & 1.000  & 1.000  \\    
    \hline \hline
    \end{tabular}
\end{table}


\section{Supplementary material for the empirical analysis in Section \ref{sec:empir}}\label{sec:app_empir}

This section provides supplementary material for the empirical analysis.
Table \ref{tab:desctab} presents the definitions of the variables used and their summary statistics.
The estimated functional coefficients and the 95\% confidence intervals are shown in Figure \ref{fig:beta}. 

\begin{table}[ht]
    \caption{Descriptive statistics ($n = 1883$)}
    \label{tab:desctab}
    \centering
    \begin{tabular}{l|cccc}
    \hline\hline
    Variable & Mean	& Std. Dev. & Min. & Max. \\
    \hline
    Landprice & 10.083 & 1.191 & 7.313 & 14.883 \\ 
    Unemployment & 3.687 & 1.134 & 0 & 10.635 \\
    Agriculture & 0.078 & 0.081 & 0 & 0.467 \\   
    Sales & 10.500 & 2.196 & 0 & 17.666 \\
    Beds & 1.089 & 1.051 & 0 & 13.489 \\ 
    Childcare & 0.290 & 0.209 & 0.000 & 2.833 \\ 
    $q(0.25)$ & 31.801 & 5.638 & 17.904 & 58.162 \\ 
    $q(0.5)$  & 52.777 & 5.865 & 38.113 & 71.003 \\
    $q(0.75)$ & 69.895 & 3.877 & 54.193 & 99.028 \\
    \hline\hline
    \end{tabular}

    \bigskip

    \begin{minipage}{0.9\textwidth}
    Definitions: Landprice = log(average residential landprice (JPY/m2)); Unemployment = unemployment rate (\%); Agriculture = proportion of agriculture, forestry, and fishery workers; Sales = log(annual commercial sales (million JPY) + 1); Beds = 100 $\times$ \# of hospital beds/population; Childcare = 1000 $\times$ \# of childcare facilities/population.
    \end{minipage}  
\end{table}

\begin{figure}[ht]
    \begin{minipage}[b]{0.48\linewidth}
    \centering
    \includegraphics[width=\textwidth]{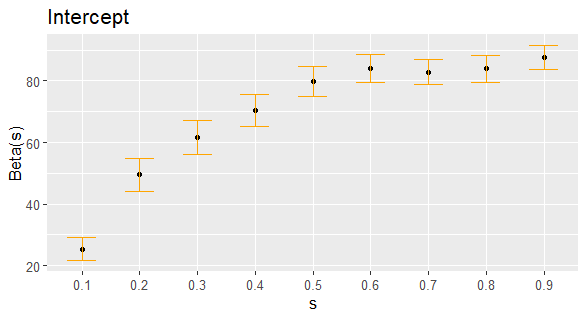}
    \end{minipage}
    \begin{minipage}[b]{0.48\linewidth}
    \centering
    \includegraphics[width=\textwidth]{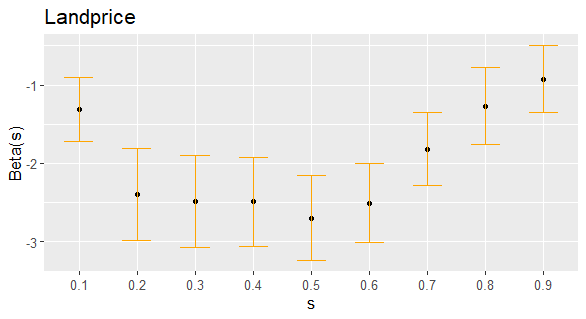}
    \end{minipage}

    \begin{minipage}[b]{0.48\linewidth}
    \centering
    \includegraphics[width=\textwidth]{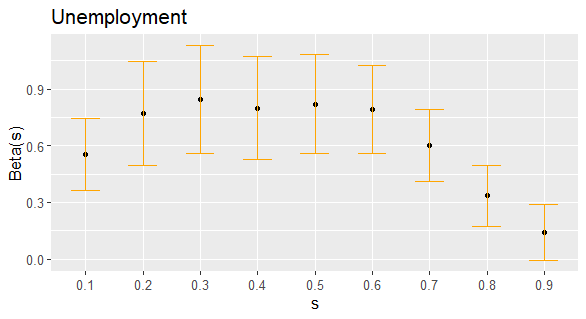}
    \end{minipage}
    \begin{minipage}[b]{0.48\linewidth}
    \centering
    \includegraphics[width=\textwidth]{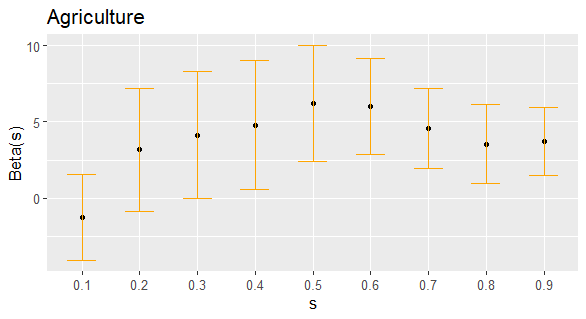}
    \end{minipage}

    \begin{minipage}[b]{0.48\linewidth}
    \centering
    \includegraphics[width=\textwidth]{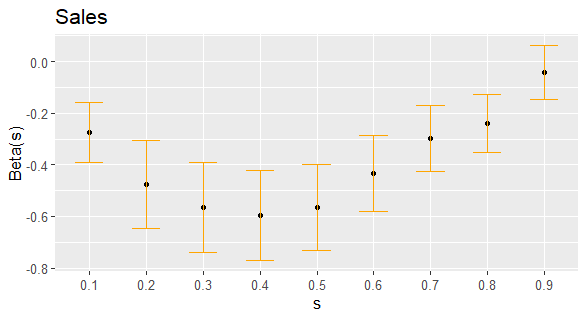}
    \end{minipage}
    \begin{minipage}[b]{0.48\linewidth}
    \centering
    \includegraphics[width=\textwidth]{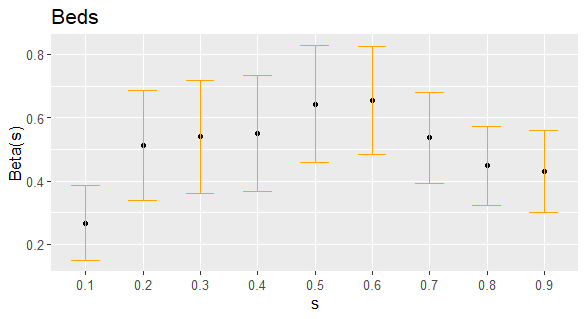}
    \end{minipage}

    \begin{minipage}[b]{0.48\linewidth}
    \centering
    \includegraphics[width=\textwidth]{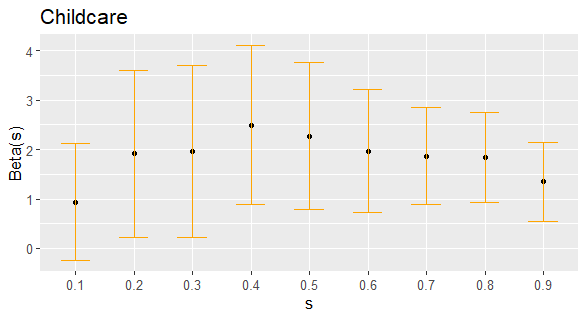}
    \end{minipage}

    \caption{Estimated coefficients}
    \label{fig:beta}
\end{figure}

\end{document}